\documentclass[journal,12pt,onecolumn,draftcls]{IEEEtran}

\usepackage{color}
\definecolor{gray}{rgb}{0.5,0.5,0.5}

\usepackage{amsmath,amsthm}
\usepackage{amssymb,amsfonts}
\usepackage{graphicx}
\usepackage{url}
\usepackage{cleveref}
\usepackage{subfigure}

\newcommand{\Mk}{\mathbf{M}_k}

\newcommand{\sumK}{\sum\limits_{k =1}^K}

\newtheorem{theorem}{Theorem}

\newtheorem{corollary}{Corollary}

\newcommand{\Exp}{{\mathbb{E}}}
\newcommand{\Expect}[2]{\Exp_{#1}\left\lbrace #2 \right\rbrace}

\newcommand{\diag}[1]{\mathrm{diag}\left(#1\right)}

\newcommand{\expb}[1]{ \exp \left\lbrace  #1 \right\rbrace  } 
\newcommand{\Expb}[1]{ \Exp_{p_n} \left\lbrace  #1 \right\rbrace  } 

\newcommand{\bb}{\mathbf{b}}
\newcommand{\bc}{\mathbf{c}}
\newcommand{\bd}{\mathbf{d}}
\newcommand{\be}{\mathbf{e}}

\newcommand{\bx}{\mathbf{x}}

\newcommand{\bA}{\mathbf{A}}
\newcommand{\bB}{\mathbf{B}}
\newcommand{\bC}{\mathbf{C}}
\newcommand{\bD}{\mathbf{D}}

\newcommand{\bF}{\mathbf{F}}
\newcommand{\bG}{\mathbf{G}}
\newcommand{\bH}{\mathbf{H}}
\newcommand{\bI}{\mathbf{I}}

\newcommand{\bM}{\mathbf{M}}
\newcommand{\bN}{\mathbf{N}}

\newcommand{\bP}{\mathbf{P}}
\newcommand{\bQ}{\mathbf{Q}}

\newcommand{\bT}{\mathbf{T}}
\newcommand{\bU}{\mathbf{U}}
\newcommand{\bV}{\mathbf{V}}
\newcommand{\bW}{\mathbf{W}}
\newcommand{\bX}{\mathbf{X}}
\newcommand{\bY}{\mathbf{Y}}
\newcommand{\bZ}{\mathbf{Z}}

\newcommand{\bzero}{\mathbf{0}}
\newcommand{\bone}{\mathbf{1}}

\newcommand{\bSigma}{{\boldsymbol\Sigma}}
\newcommand{\blambda}{{\boldsymbol\lambda}}
\newcommand{\bLambda}{{\boldsymbol\Lambda}}
\newcommand{\bPhi}{{\boldsymbol\Phi}}

\newcommand{\bOmega}{{\boldsymbol\Omega}}

\newcommand{\bPi}{{\boldsymbol\Pi}}

\newcommand{\bGamma}{{\boldsymbol\Gamma}}

\usepackage{algorithm}
\usepackage{algorithmicx}
\usepackage{algpseudocode}
\usepackage{cite}

\ifCLASSINFOpdf
\else
\fi
\begin{document}
%
\title{2D Beam Domain Statistical CSI Estimation for Massive MIMO Uplink}

\author{An-An~Lu,~\IEEEmembership{Member,~IEEE}, Yan Chen, ~Xiqi~Gao,~\IEEEmembership{Fellow,~IEEE} 
\thanks{A.-A. Lu and X. Q. Gao are with the National Mobile Communications Research Laboratory (NCRL), Southeast University,
Nanjing, 210096 China, and also with Purple Mountain Laboratories, Nanjing 211111, China, e-mail: aalu@seu.edu.cn, xqgao@seu.edu.cn.}
\thanks{Y. Chen is with the National Mobile Communications Research Laboratory (NCRL), Southeast University,
Nanjing, 210096 China, }
}


\maketitle

\begin{abstract}
In this paper, we investigate the beam domain statistical channel state information (CSI) estimation for the two dimensional (2D) beam based statistical channel model (BSCM) in massive MIMO systems.  The problem is to estimate the beam domain channel power matrices (BDCPMs) based on multiple receive pilot signals. A receive model shows the relation between the statistical property of the receive pilot signals and the BDCPMs is derived from the 2D-BSCM. On the basis of the receive model, we formulate an optimization problem with the Kullback-Leibler (KL) divergence.
By solving the optimization problem, a novel method to estimate the statistical CSI without involving instantaneous CSI is proposed. The proposed method has much lower complexity than the MMV focal underdetermined
system solver (M-FOCUSS) algorithm. We further reduce the complexity of the proposed method by utilizing the circulant structures of particular matrices in the algorithm. We also showed the generality of the proposed method by introducing another application.
Simulations results show that the proposed method works well and bring significant performance gain when used in channel estimation.
\end{abstract}


%
\IEEEpeerreviewmaketitle

\newpage
\section{Introduction}
Massive multiple-input multiple-output (MIMO) \cite{marzetta2016fundamentals, bjornson2014massive, lu2014overview, clerckx2016rate, sanguinetti2019toward} has been one of the key enabling technologies of the fifth generation (5G) wireless communications networks. It provides enormous capacity gains and achieves high energy efficiency by employing a large number of antennas at the base station (BS). In massive MIMO systems, multi-user MIMO (MU-MIMO) \cite{liu2012downlink} transmissions on the same time and frequency resource are enhanced significantly. Furthermore, massive MIMO also brings many new applications and services \cite{barneto2021full, zheng2020joint}. For the antenna array equipped in the BS, the uniform planar array (UPA) is widely used in practical massive MIMO systems since it has compact size. In this paper, we investigate the three dimensional (3D) massive MIMO systems equipped with UPA.

For massive MIMO systems, the beam based statistical channel model (BSCM) \cite{yu2021hf, wang2021robust, lu2020robust, lu2021iterative} is used in the literature for robust linear precoder design and system performance analysis.
The BSCM is extended from the unitary–independent–unitary (UIU)  model \cite{tulino2005impact} or the jointly correlated channel model \cite{sayeed2002deconstructing, weichselberger2006stochastic, sunbeam} with the eigen-matrices being replaced by the oversampled discrete Fourier transform (DFT)  matrix.
It is more accurate than the beam domain channel model based on the DFT based beams. Thus, the linear precoder design based on the BSCM can achieve significant performance gain compared with that based on the beam domain channel model as shown in \cite{lu2021iterative, lu2020robust}. The model can also be extended to the angle-delay domain as a two dimensional (2D) BSCM, which also brings performance gain in the channel estimation \cite{Yang2022channel}. To achieve these performance gains in massive MIMO systems, the statistical parameters in the channel model need to be known in advance. Although  the statistical parameters for the BSCM is very important, the problem of estimating them is not mentioned in those works.  
Due to its importance, we consider the problem of estimating the statistical CSI for the 2D-BSCM based on the receive pilot signals in this paper.

In the literature, the statistical CSI is often obtained based on the estimated instantaneous CSI \cite{sunbeam}, or obtained  through the 
expectation–maximization (EM) algorithms \cite{wen2014channel} which iteratively estimate the instantaneous and statistical CSI. There also exists works \cite{anjinappa2020off} that obtain the covariance matrix directly without instantaneous CSI being involved. For the 2D-BSCM, the problem becomes to obtain the channel power matrices in the beam domain, which has not been addressed in the literature. In the 2D-BSCM, the angle-delay domain or the beam domain channel coefficient is sparse due to the limit number of resolvable multi-pathes. When there exists no noise, the considered problem can also be viewed as a multiple measurement vectors (MMV) problem \cite{cotter2005sparse}, which is a classic compressive sensing problem. The MMV focal underdetermined
system solver (M-FOCUSS) algorithm \cite{cotter2005sparse} can be applied to solve the problem and obtain the instantaneous CSI, which are then used to compute the statistical CSI.

However, the noise can not be ignored in the practical massive MIMO  systems. Furthermore, the computational complexity of the M-FOCUSS method is not satisfied since it needs to compute the instantaneous CSI first and the dimension for the MMV problem in the considered massive MIMO is very high.  Furthermore, the statistical CSI of the 2D-BSCM can also be utilized to improve the estimating performance of instantaneous CSI in practical massive MIMO systems. Thus, it is better to obtain the statistical CSI for the considered problem before estimating the instantaneous CSI. In conclusion, we need a new method with lower complexity to estimate the statistical CSI for the 2D-BSCM.

To achieve this goal, we first derive a theorem which gives the relation between the statistical property of the consider channel matrices and the beam domain channel power matrices (BDCPMs). Then, we derive a receive model for the BDCPMs based on this relation. Based on the derived model, the statistical parameters of the 2D-BSCM can be estimated directly without involving the instantaneous CSI. To estimate the BDCPM, we then  formulate a new optimization problem based on the Kullback-Leibler (KL) divergence. By solving the problem, we propose a novel method to obtain the BDCPM for the 2D-BSCM. The proposed algorithm has much lower complexity than that of the M-FOCUSS method. Furthermore, we further reduce the complexity of the proposed method by utilizing the circulant structures of certain matrices in the algorithm. We also show the generality of the proposed method by presenting another application.

The main contributions of this paper are summarized as follows:
\begin{enumerate}
\item We derive a receive model for the BDCPM of the 2D-BSCM. The receive model can be used to estimate the statistical CSI directly without involving the instantaneous CSI.
\item We propose a novel method to obtain the BDCPMs based on the receive model and the KL divergence. Compared with the M-FOCUSS method, the proposed method has much lower complexity.
\item We further reduce the complexity of the proposed method by utilizing the circulant structure of certain matrices in the algorithm.
\end{enumerate}

The rest of this article is organized as follows. Section II introduces the system model and formulates the problem. Section III presents the estimation of beam domain power matrices. Section IV provides simulation results. Section V draws the conclusion. The proofs of the theorems and corollary are provided in the Appendices.

\subsection{Notations}
Throughout this paper, uppercase boldface letters and lowercase boldface letters are used for matrices and vectors, respectively. The superscripts $(\cdot)^*$, $(\cdot)^T$ and $(\cdot)^H$ denote the conjugate, transpose and conjugate transpose operations, respectively. The operator ${\mathbb E}\{\cdot\}$ denotes the mathematical expectation operator. In some cases, where it is not clear, we will employ subscript to emphasize the definition. The operators $\odot$ and $\otimes$  represent the Hadamard and Kronecker product, respectively.  
We use $\bzero_{N\times M}$ and $\bone_{N,M}$ to denote $N\times M$ matrices or vectors of all zeros and all ones, respectively.
The $N \times N$ identity matrix is denoted by $\mathbf{I}_N$, and $\bI_{N\times M}$ is used to denote $[\bI_N~\bzero_{N\times(M-N)} ]$ when $N<M$ and $[\bI_M~\bzero_{M\times(N-M)} ]^T$ when $N>M$. The subscripts of $\bI,\bzero$ and $\bone$ can sometimes be omitted for convenience. 
We use $[\mathbf{A}]_{ij}$ to denote the $(i,j)$-th entry of the matrix $\mathbf{A}$. The operators ${\rm{tr}}(\cdot)$ and $\det(\cdot)$ represent the matrix trace and determinant, respectively. We use $\diag{\bX}$ to denote a column vector composed of the main diagonal elements of a square matrix $\bX$, and $\diag{\bx}$ to denote the diagonal matrix with $\bx$ along
its main diagonal. A $N$-dimensional normalized DFT matrix is denoted as $\bF_{N}$. We further define the permutation matrix as
\begin{align}
\bPi_N^n=\begin{bmatrix}
\bzero & \bI_{N-n} \\
\bI_n & \bzero
\end{bmatrix},
\end{align}
where $0\le n\le N$.

\section{System Model and Problem Formulation}

\subsection{System Model}
We consider a 3D massive MIMO system with frequency selective fading channels.
The system consists of one BS equipped with a UPA array and $K$ user terminals (UTs) with single antennas. The number of the antennas at the BS is $M_r$, where the numbers of antennas at each
vertical column and horizontal row are $M_{r,z}$ and $M_{r,x}$, respectively. 
The orthogonal frequency division multiplexing (OFDM) \cite{stuber2004broadband} modulation is used to transform the frequency selective fading channel into multiple of parallel channels. Thus, the considered system is a massive MIMO-OFDM system. The number of subcarriers in the massive MIMO-OFDM system is $M_c$, and $M_p$ subcarriers are used for uplink pilot signal transmission. The
length of the cyclic prefix (CP) and the sampling interval are denoted  as $M_g$ and $T_s$.

We restrict our considerations to stationary channels and use the BSCM to describe the spatial-temporal correlations of each channel. 
We denote the polar and azimuthal angles of arrival at the BS by $\theta_r,\phi_r$. Let $d_z$ and $d_x$ be the antenna spacing of each row and each column of the UPA. Let $\Delta_z = \frac{d_z}{\lambda}$, $\Delta_x = \frac{d_x}{\lambda}$, and $u_r$ and $v_r$ denote the directional cosines with respect to the $z$ axis and $x$ axis, respectively. Then, we have  $u_r=\cos\theta_r$ and $v_r=\sin\theta_r\cos\phi_r$. 
The steering vector at the BS side is given by 
\begin{equation}
        \mathbf{a}_r(u_r, v_r)=  \mathbf{v}_z(u_r) \otimes \mathbf{v}_x(v_r),
        \label{eq:vector_a_kronecker_product_form}
\end{equation}
where
\begin{align}
        \mathbf{v}_z(u_r)= [1 ~~ e^{-j2\pi\Delta_z  {u_r}} ~~  \cdots ~~ e^{-(M_{z}-1)j2\pi \Delta_z  {u_r}}]^T, \\
        \mathbf{v}_x(v_r)= [1 ~~ e^{-j2\pi \Delta_x {v_r}} ~~  \cdots ~~ e^{-(M_{x}-1)j2\pi \Delta_x  {v_r}}]^T.
\end{align}
In this paper, both $d_z$ and $d_x$ are assumed equal to $\frac{1}{2}\lambda$. Then, we obtain that $\Delta_z= \Delta_x=\frac{1}{2}$. 
Let $\mathbf{V}$ be the matrix of sampled steering vectors defined as
\begin{align}
\mathbf{V} = \mathbf{V}_z \otimes \mathbf{V}_x \in \mathbb{C}^{M_{r} \times N_{r}},
\end{align}
where
\begin{align}
    \mathbf{V}_z&=[\mathbf{v}_z(u_{r,1}) ~~ \mathbf{v }_z(u_{r,2}) ~~ \cdots ~~ \mathbf{v}_z(u_{r,N_z})]\in \mathbb{C}^{M_{r,z} \times N_{r,z}}, \\
    \mathbf{V}_x&=[\mathbf{v}_x(v_{r,1}) ~~ \mathbf{v }_x(v_{r,2}) ~~ \cdots ~~ \mathbf{v}_x(v_{r,N_x})]\in \mathbb{C}^{M_{r,x} \times N_{r,x}}.
\end{align}
We define $N_{a,z}=\frac{N_z}{M_z}$ and $N_{a,x}=\frac{N_x}{M_x}$ as the vertical and horizontal angle domain fine factors (FFs), respectively. The matrices $\bV_z$ and $\bV_x$ are oversampled DFT matrices when $N_{a,z}$ and $N_{a,x}$ are integers, and the directional cosines are uniformly sampled in the range of -1 to 1. Then, $\bV_z$ and $\bV_x$ can be represented as $\bV_z=\bI_{M_z,N_z}\bF_{N_z}$ and $\bV_x=\bI_{M_x,N_x}\bF_{N_x}$, respectively.

The frequency basis vector $\mathbf{b}_r(\tau_r) \in \mathbb{C}^{M_p \times 1}$ is given by
\begin{equation}
        \mathbf{b}_r(\tau_r)= [1 ~~ e^{-j2\pi \Delta_f{\tau_r}} ~~ ~~ \cdots ~~ e^{-(M_{p}-1)j2\pi \Delta_f{\tau_r}}]^T,
\end{equation}
where $\Delta_f=\frac{1}{M_cT_s}$ is the frequency spacing between neighboring carriers and $\tau_r$ is the delay.  
We define the matrix of sampled frequency basis vectors as
    \begin{equation}
        \mathbf{U}=[\mathbf{b}_r(\tau_{r,1}) ~~ \mathbf{b}_r(\tau_{r,2}) ~~ \cdots ~~ \mathbf{b}_r(\tau_{r,N_p})] \in \mathbb{C}^{M_{p} \times N_{p}}.
    \end{equation}
The delay domain fine factor is defined as $N_{a,p}=\frac{N_p}{M_p}$. Similarly, $\bU$ is an oversampled DFT matrix when $N_{a,p}$ is an integer and $\tau_r$ is uniformly sampled as $\tau_{r,\ell}=\frac{\ell-1}{N_p\Delta_f}$. In this case, we have that $\bU=\bI_{M_p,N_p}\bF_{N_p}$. We assume that the delay is within the guard interval, i.e., $\tau_r\le M_gT_s$. We define $M_f=\lceil\frac{M_pM_g}{M_c}\rceil$ and $N_f=N_{a,p}M_f$, where the notation $\lceil\cdot\rceil$ represents rounding upwards. Then, the vector set $\left\{ \bb_r(\tau_{r,1}),\bb_r(\tau_{r,2}),...,\bb_r(\tau_{r,N_{f}}) \right\}$ is enough to contain all the sampled frequency basis vectors since $\frac{N_f}{N_p\Delta_f}\ge M_gT_s$. Thus, we further define $\bU_f\in \mathbb{C}^{M_p\times N_f}$ as
\begin{equation}
\mathbf{U}_f=\bU\bI_{N_p,N_f}=[\mathbf{b}_r(\tau_{r,1}) ~~ \mathbf{b}_r(\tau_{r,2}) ~~ \cdots ~~ \mathbf{b}_r(\tau_{r,N_f})].
\end{equation}


By using the 2D-BSCM, the space-frequency domain channel matrix between the $k$-th UE and the BS for the $t$-th OFDM symbol can be modeled as \cite{yu2021hf, wang2021robust, lu2020robust, lu2021iterative}
    \begin{equation}
        \mathbf{H}_{k,t}=\mathbf{V}(\mathbf{M}_k\odot\mathbf{W}_{k,t})\mathbf{U}_f^T,
        \label{eq:channel_matrix_correlation_model}
    \end{equation}
where the matrix $\mathbf{M}_k$ is an $N_r \times N_f$  deterministic matrix and remains unchanged in different OFDM symbols, and $\mathbf{W}_{k,t}$ is a complex Gaussian random vector consisting of independent and identically distributed (i.i.d.) elements with zero mean and unit variance.   We also assume that $\bW_{k,t}$ and $\bW_{k',t}$ are independent of each other when $k\ne k'$. We define $\mathbf{G}_{k,t}=\mathbf{M}_k\odot\mathbf{W}_{k,t}$, which  is the angle-delay domain channel matrix and also called the two dimensional beam domain channel matrix (2D-BDCM). The  2D-BDCPM of the $k$-th user is defined as $\mathbf{\Omega}_k=\mathbf{M}_k\odot\mathbf{M}_k$, which is a sparse matrix since most of the channel power is distributed in a limited number of resolvable spatial directions and time delays.


\subsection{Problem Formulation} 

In the two dimensional channel model \eqref{eq:channel_matrix_correlation_model}, the matrix $\mathbf{U}_f$ and $\mathbf{V}$ are deterministic matrices. The unknown statistical parameter is the matrix $\Mk$ or equivalently the 2D-BPCMs $\boldsymbol{\Omega}_k$.
The statistical CSI or the matrix $\boldsymbol{\Omega}_k$ can be exploited to schedule UTs and improve the estimation
performance of instantaneous CSI, which will bring significant system performance gain. Thus,
it is very important obtain the statistical CSI or the matrix $\boldsymbol{\Omega}_k$.

To estimate the matrices $\boldsymbol{\Omega}_k$, we can use the received pilot signals.
We now consider the statistical CSI estimation based on the uplink pilot transmission. We use the pilot signal sequence in \cite{3GPP} as
\begin{align}
\bx_{q,p} = \tilde{\bx}_q \odot \bb_r(\tau_{r,pN_f}),
\end{align}
where $q$ and $p$ denote the root coefficient and cyclic shift, respectively. The sequence $\tilde{\bx}_q$ is the Zadoff-Chu (ZC) sequence with root coefficient $q$, which is specifically represented as \cite{3GPP}
\begin{align}
[\tilde{\bx}_q]_n=\expb{-j\dfrac{\pi q n(n+1)}{N_l}},\ n=0,...,M_p-1,
\end{align}
where $N_l$ is the largest prime number such that $N_l<M_p$.
In particular, the pilots degenerate into orthogonal pilots (OPs) when there is only one root. However, the overhead of orthogonal pilots is relatively large. There might not be enough pilot resource for the orthogonal pilots as the number of the users increase in the massive MIMO systems. Thus, we use the nonorthogonal pilots in  \cite{3GPP} to schedule more UTs in a OFDM symbol. We denote the number of roots and the number of UTs on the $q$-th root as $Q$ and $P_{q}$. Let the matrices $\mathbf{X}_{q,p},\ {\tilde{\mathbf{X}}_q}$ and ${\bB_r(\tau)}$ denote $\diag{\bx_{q,p}},\ \diag{\tilde{\mathbf{x}}_q}$ and  $\diag{\bb_r(\tau)}$, respectively. We use the subscript $q,p$ to replace $k$ for convenience.
The received pilot signal $\mathbf{Y}_t \in \mathbb{C}^{M_r \times M_p}$ at the BS for the $t$-th OFDM symbol is given by
    \begin{equation}\label{eq:receiver_model1}
        \mathbf{Y}_t= \sumK \mathbf{H}_{k,t}\mathbf{X}_{k} + \mathbf{Z}_t = \sum_{q=1}^{Q} \left( \sum_{p=1}^{P_{q}}\bH_{q,p,t}\bB_r(\tau_{r,pN_f}) \right)\tilde{\bX}_{q} + \mathbf{Z}_t,
    \end{equation}
where $\mathbf{Z}_t$ is a complex Gaussian noise matrix consisting of independent and identically distributed (i.i.d.) elements with zero mean and variance $\sigma_z^2$. Substituting the channel model (\ref{eq:channel_matrix_correlation_model}) into (\ref{eq:receiver_model1}), we can obtain that
\begin{align} 
\mathbf{Y}_t &= \sum_{q=1}^{Q} \left( \sum_{p=1}^{P_{q}}\bV\bG_{q.p,t}\bU_f^T\bB_r(\tau_{r,pN_f}) \right)\tilde{\bX}_{q} + \mathbf{Z}_t \notag \\ 
&= \sum_{q=1}^{Q} \left( \sum_{p=1}^{P_{q}}\bV\bG_{q,p,t}\bI_{N_f,N_p}\bU^T\bB_r(\tau_{r,pN_f}) \right)\tilde{\bX}_{q} + \mathbf{Z}_t \notag\\ 
&\overset{\small{(a)}}{=} \sum_{q=1}^{Q}  \bV \left(\sum_{p=1}^{P_{q}}\bG_{q,p,t}\bI_{N_f,N_p}\bPi_{N_p}^{(p-1)N_f}\right)   \bU^T  \tilde{\bX}_{q} + \mathbf{Z}_t,  
\end{align}
where the proof of step $(a)$ can be performed
in a similar way as in \cite{You16Channel}. Let the matrix $\tilde{\bG}_{q,t}$ be defined as  
\begin{equation}
\tilde{\bG}_{q,t} = \sum_{p=1}^{P_{q}}\bG_{q,p,t}\bI_{N_f,N_p}\bPi_{N_p}^{(p-1)N_f}.
\end{equation}
The receive model becomes
\begin{align} 
\mathbf{Y}_t = \sum_{q=1}^{Q}  \bV \tilde{\bG}_{q,t}   \bU^T  \tilde{\bX}_{q} + \mathbf{Z}_t, 
\end{align}
and we can also obtain that
\begin{align}
\tilde{\bG}_{q,t}=\left[ \bG_{q,1,t}~\bG_{q,2,t}~\cdots~\bG_{q,P_q,t}~ \bzero \right]\in\mathbb{C}^{N_r\times N_p}.
\end{align}
Let the matrices $\mathbf{G}_t$ and $\mathbf{P}$ be defined as
\begin{align}
&\mathbf{G}_t=[\tilde{\mathbf{G}}_{1,t} ~ \tilde{\mathbf{G}}_{2, t} ~ \cdots ~ \tilde{\mathbf{G}}_{Q, t} ],\\  
&\mathbf{P}=[\tilde{\mathbf{X}}_{1}^T\bU ~ \tilde{\mathbf{X}}_{2}^T\bU ~ \cdots ~ \tilde{\mathbf{X}}_{Q}^T\bU ]^T.
\end{align} 
We can rewrite the receive model as
\begin{align}\label{eq:receive model}
\mathbf{Y}_t  
 = \bV\bG_t\bP + \mathbf{Z}_t.
\end{align}
For convenience, we also define  
\begin{align}
&\tilde{\bM}_q = \left[ \bM_{q,1}~~\bM_{q,2}~~\cdots~~\bM_{q,P_q}~~ \bzero \right]\in\mathbb{C}^{N_r\times N_p}, \\
&\tilde{\bOmega}_q = \left[ \bOmega_{q,1}~~\bOmega_{q,2}~~\cdots~\bOmega_{q,P_q}~~ \bzero \right]\in\mathbb{C}^{N_r\times N_p}, \\
&\bM = \left[ \tilde{\bM}_1~~\tilde{\bM}_2~~\cdots~~\tilde{\bM}_Q \right], \\
&\bOmega = \left[ \tilde{\bOmega}_1~~\tilde{\bOmega}_2~~\cdots~~\tilde{\bOmega}_Q \right].
\end{align}

In the receive model \eqref{eq:receive model}, the matrix $\bV$ and $\bP$ are deterministic matrices. The unknown
statistical parameter is the matrix $\bM$ or equivalently the matrix $\boldsymbol{\Omega}$. Since the statistical CSI varies
slowly and remains the same over a short period of time, we can use multiple received pilot signals to estimate it. Thus, the problem is to estimate
$\boldsymbol{\Omega}$ base on multiple receive pilot signals $\mathbf{Y}_t$. When there is no noise, the problem can also be
viewed as a MMV problem \cite{cotter2005sparse}. Multiple instantaneous beam domain channel coefficients $\mathbf{G}_t$ can be
obtained by using the M-FOCUSS method \cite{cotter2005sparse} and then be used to obtain the statistical $\boldsymbol{\Omega}$. However,
we need to consider the noise here. Furthermore, the complexity of the M-FOCUSS method is also not satisfied because it needs to estimate a large number of instantaneous channel matrices simultaneously  and the dimension of the MMV problem becomes extremely large as the number of antennas and frequency carriers increases. Thus, we need a new method to
estimate the statistical $\boldsymbol{\Omega}$ directly.

\section{Estimation of Beam Domain Power Matrices}
\label{sec:obtain_CCM_using_CS_method_2}
\subsection{2D-BDCPM acquisition algorithm based on KL divergence minimization}
Since the statistical parameters to be estimated are on the two dimensional angle-delay domain, it is natural to convert the received pilot signals into   signals on the angle-delay domain.
By left multiplying $\mathbf{Y}_t$ with $\mathbf{V}^H$ and right multiplying it with $\mathbf{P}^H$, we obtain the received pilot signal   on the angle-delay domain  as
       \begin{equation}
        \mathbf{V}^H\mathbf{Y}_t \mathbf{P}^H=  \mathbf{V}^H\mathbf{V}\mathbf{G}_{t}\mathbf{P}\mathbf{P}^H + \mathbf{V}^H\mathbf{Z}_t \mathbf{P}^H.
    \end{equation}
 
The relation between the received pilot signal and the 2D-BDCPM is not clear yet.
To figure out their relation, we present the following theorem that shows one important property of the 2D-BDCM ${\mathbf{G}}_{t}$.
\begin{theorem}
\label{th:theorem_channel_power_matrices_relation}
The 2D-BDCM ${\mathbf{G}}_{t}$ satisfies
\begin{equation}
        \mathbb{E}\{(\mathbf{C}_1{\mathbf{G}}_{t}\mathbf{C}_2) \odot (\mathbf{C}_1{\mathbf{G}}_{t}\mathbf{C}_2) ^*\}=  \mathbf{T}_1\boldsymbol{\Omega}\mathbf{T}_2,
\end{equation}
where $\mathbf{C}_1$ and $\mathbf{C}_2$ are constant matrices, $\mathbf{T}_1=\mathbf{C}_1 \odot \mathbf{C}_1^*$ and $\mathbf{T}_2 = \mathbf{C}_2 \odot \mathbf{C}_2^*$.
\end{theorem} 
\begin{proof}
	The proof is provided in Appendix \ref{appendice: theorem_channel_power_matrices_relation}.
\end{proof}

Theorem \ref{th:theorem_channel_power_matrices_relation} provides an important property of the product of one random matrix with independent entries and two deterministic matrices. 
Let $\mathbf{\Phi}$ denote the expectation of the receive power matrix on the angle-delay domain as
\begin{equation}
    \mathbb{E}\{[(\mathbf{V}^H\mathbf{Y}\mathbf{P}^H) \odot (\mathbf{V}^H\mathbf{Y}\mathbf{P}^H)^*]\}. \nonumber
\end{equation}
From Theorem \ref{th:theorem_channel_power_matrices_relation}, we can obtain the receive model of the statistical parameter $\boldsymbol{\Omega}$ as
\begin{IEEEeqnarray}{Cl}
         \mathbf{\Phi}  = \mathbf{T}_{a}\boldsymbol{\Omega} \mathbf{T}_{f} +  \mathbf{N},
         \label{eq:statistical_CSI_equation}
\end{IEEEeqnarray}
where $\mathbf{T}_{a}$, $\mathbf{T}_{f}$ and $\mathbf{N}$ are deterministic matrices defined as 
\begin{align}
\label{eq:Ta}\mathbf{T}_{a}&=(\mathbf{V}^H\mathbf{V})  \odot (\mathbf{V}^H\mathbf{V})^*,\\
\label{eq:Tf}\mathbf{T}_{f}&=(\bP\bP^H)\odot(\bP\bP^H)^*\notag \\
&=\begin{bmatrix}
\left(\mathbf{U}^T\tilde{\bX}_{1}\tilde{\bX}_{1}^H\mathbf{U}^*\right)\odot \left(\mathbf{U}^T\tilde{\bX}_{1}\tilde{\bX}_{1}^H\mathbf{U}^*\right)^* & \cdots & \left(\mathbf{U}^T\tilde{\bX}_{1}\tilde{\bX}_{Q}^H\mathbf{U}^*\right)\odot \left(\mathbf{U}^T\tilde{\bX}_{1}\tilde{\bX}_{Q}^H\mathbf{U}^*\right)^* \\
\vdots & \ddots &\vdots \\
\left(\mathbf{U}^T\tilde{\bX}_{Q}\tilde{\bX}_{1}^H\mathbf{U}^*\right)\odot \left(\mathbf{U}^T\tilde{\bX}_{Q}\tilde{\bX}_{1}^H\mathbf{U}^*\right)^* & \cdots & \left(\mathbf{U}^T\tilde{\bX}_{Q}\tilde{\bX}_{Q}^H\mathbf{U}^*\right)\odot \left(\mathbf{U}^T\tilde{\bX}_{Q}\tilde{\bX}_{Q}^H\mathbf{U}^*\right)^* 
\end{bmatrix} , 
\\
\mathbf{N}&=(\mathbf{V}^H\odot  \mathbf{V}^T)\mathbb{E}\{\mathbf{Z} \odot \mathbf{Z}^*\}(\mathbf{P}^H\odot  \mathbf{P}^T) = {M_tM_p\sigma_z^2}\bone.
\end{align} 

To estimate the 2D-BDCPM, i.e., the matrix $\boldsymbol{\Omega}$, from $\mathbf{\Phi}$, we can solve \eqref{eq:statistical_CSI_equation} directly to obtain an exact solution.
However, it is not an easy task since the elements of $\boldsymbol{\Omega}$ need be nonnegative. Even more,  \eqref{eq:statistical_CSI_equation}  might not have an exact solution.
Thus, we propose to construct an optimization problem to obtain approximate solution.
To achieve this goal, we need an objective function first.
This means that we need to choose a type of divergence or distance  between the matrix $\mathbf{\Phi}$ and the sum $\mathbf{T}_{a}\boldsymbol{\Omega}\mathbf{T}_{f} +  \mathbf{N}$ according to \eqref{eq:statistical_CSI_equation}. Since both the matrices are matrices with positive real elements, the Kullback--Leibler (KL) divergence between sequences of positive real elements is a natural option and used as the objective function.

Meanwhile, optimizing $\boldsymbol{\Omega}$ directly is still complicate since it has the constraint that its elements
are nonnegative.  
From the relation $\boldsymbol{\Omega}=\bM \odot \bM$, we know that estimating the matrix $\bM$ is equivalent to
estimating the matrix $\boldsymbol{\Omega}$. Thus, we choose to optimize $\mathbf{M}$ instead of $\boldsymbol{\Omega}$ because it has no constraint. 
We define the function $f(\bM)$ as the KL divergence between the matrices $\mathbf{\Phi}$ and $\mathbf{T}_{a}\boldsymbol{\Omega}\mathbf{T}_{f}+\mathbf{N}$, \textit{i.e.}, \cite{amari2016information}
\begin{align}
f(\bM) &= \sum_{ij}[\mathbf{\Phi}]_{ij}\log\frac{[\mathbf{\Phi}]_{ij}}{ [\mathbf{T}_{a}\boldsymbol{\Omega}\mathbf{T}_{f}+ \mathbf{N}]_{ij}}   + \sum_{ij} [\mathbf{T}_{a}\boldsymbol{\Omega}\mathbf{T}_{f}+ \mathbf{N}]_{ij}   - \sum_{ij} [\mathbf{\Phi}]_{ij}.
\end{align}
Using the KL divergence $f(\bM)$, we are now able to formulate an unconstrained optimization problem as
\begin{align}	\label{eq:optimization_problem_of_channel_power_matrices}
	\bM^{\star} &= \arg\min\limits_{\bM} f(\bM).
\end{align}

To solve the optimization problem in \eqref{eq:optimization_problem_of_channel_power_matrices},  the gradient method \cite{bertsekas1997nonlinear} can be used. Thus, we calculate the gradient of $f(\bM) $ with respect to $\bM$ first.
In the following theorem, we provide the gradients of two items in $f(\bM)$.
\begin{theorem}
	\label{th:theorem_of_gradients}
	The gradients of  two items in $f(\bM)$ can be obtained as 
	\begin{align}\label{eq:gradients of two items}
	&\frac{\partial\sum_{ij}  [\mathbf{T}_{a}(\mathbf{M} \odot \mathbf{M})\mathbf{T}_{f}]_{ij}}{\partial {\mathbf{M}}}
	= 2\left( \mathbf{T}_{a}\bone\mathbf{T}_{f} \right) \odot  \mathbf{M},
	\\
	& \frac{\partial \sum_{ij}[\mathbf{\Phi}]_{ij}\log \left[\mathbf{T}_{a}(\mathbf{M} \odot \mathbf{M})\mathbf{T}_{f}+ \mathbf{N} \right]_{ij}}{ \partial \mathbf{M}}
	=  2\left(\mathbf{T}_{a}\mathbf{Q}\mathbf{T}_{f}\right) \odot \mathbf{M},
	\end{align}	 
where the matrix $\mathbf{Q}$ is the Hadmard division of two matrices with the elements being defined as 
\begin{equation}
[\mathbf{Q}]_{ij}=\frac{[\mathbf{\Phi}]_{ij}}{[\mathbf{T}_{a}\boldsymbol{\Omega}\mathbf{T}_{f}+  \mathbf{N}]_{ij}}.
\end{equation}
\end{theorem}
\begin{proof}
	The proof is provided in Appendix \ref{appendix:theorem of gradients}.
\end{proof}

From Theorem \ref{th:theorem_of_gradients} and the gradients of  other items in $f(\bM)$ with respect to $\bM$ are zeros,  we obtain the gradient of the function $f(\bM)$ as
\begin{align}\label{eq:gradient of f}
	 \frac{\partial f(\bM)}{\partial \bM} &= 2(\mathbf{T}_{a}\bone\mathbf{T}_{f}) \odot  \mathbf{M} - 2(\mathbf{T}_{a}\mathbf{Q}\mathbf{T}_{f}) \odot \mathbf{M} \notag \\
	 &= 2(\mathbf{T}_{a}(\bone-\mathbf{Q})\mathbf{T}_{f}) \odot  \mathbf{M}.
\end{align}
With the obtained gradient, we can apply the gradient method to obtain 
the optimal $\mathbf{M}$ as
\begin{IEEEeqnarray}{Cl}
	\mathbf{M}^{d+1}  = \mathbf{M}^{d}  - \delta^d \frac{\partial f(\bM^d)}{\partial \bM^d},
	\label{eq:iterative_equation_to_obtain_Mk}
\end{IEEEeqnarray}    
where the superscript $d$ represents the iteration number and $\delta^d$ is the step size which can be obtained by the line search method \cite{boyd2009convex}. Recall that $\mathbf{\Phi}$ denotes $\mathbb{E}\{[(\mathbf{V}^H\mathbf{Y}\mathbf{P}^H) \odot (\mathbf{V}^H\mathbf{Y}\mathbf{P}^H)^*]\}$.
Thus, it is not possible to obtain the matrix $\mathbf{\Phi}$ directly in practice. Instead, we use the sample
average $\sum_{t=1}^T\frac{1}{T}\{[(\mathbf{V}^H\mathbf{Y}_t\mathbf{P}^H) \odot (\mathbf{V}^H\mathbf{Y}_t\mathbf{P}^H)^*]\}$, where $T$ is the number of samples.  The obtained algorithm is summarized as Algorithm \ref{alg:KL}.
\begin{algorithm}[htb]
	\caption{2D-BDCPM acquisition algorithm based on KL divergence minimization}
	\label{alg:KL}
	\begin{algorithmic}[1]
		\State Use $T$ received pilot signals to calculate  $\bPhi$ as $\frac{1}{T}\sum_{t=1}^T{ (\bV^H\bY_t\bP^H)\odot(\bV^H\bY_t\bP^H)^* }$
		\State \textbf{Initialization}: set $d=0$ and the maximum number of iterations as $D$, select appropriate $\delta^0$, $\delta_{min}<\delta^0$ and $\alpha\in (0,1)$, and initialize $\bOmega^0=\frac{1}{QN_rN_p}\bPhi$ and $\bM^0=\sqrt{\bOmega}$
		\State \textbf{Repeat} 
		\State \qquad Calculate the gradient of $ f(\bM^d)$ as $
		\frac{\partial f(\bM^d)}{\partial \bM^d} = 2(\mathbf{T}_{a}(\bone-\mathbf{Q})\mathbf{T}_{f}) \odot  \mathbf{M}^d $
		\State \qquad \textbf{while} $\delta^d>\delta_{min}$ 
		\State \qquad\qquad Update $\mathbf{M}^{d+1}  = \mathbf{M}^{d}  - \delta^d \frac{\partial f(\bM^d)}{\partial \bM^d}$
		\State \qquad\qquad \textbf{if} $f(\bM^{d+1})\ge f(\bM^{d})$ {then} $\delta^d=\alpha\delta^d$ and $\bM^{d+1}=\bM^d$
		\State \qquad\qquad \textbf{else} {break}
		\State \qquad Set $d = d+1$. 
		\State \textbf{until} $\delta^d\le \delta_{min}$ or $d=D$
		\State Calculate $\bOmega^d=\bM^d\odot\bM^d$
	\end{algorithmic}
\end{algorithm}

Since the computational complexity of  products is much higher than that of  additions, we use the number of the complex products as the computational complexity. Then, the complexity of Algorithm 1 is
dominate by the matrix product $\bT_a\bQ\bT_f$ in (\ref{eq:gradient of f}), whose complexity is $O\left(QN_rN_p(N_r+QN_p)\right)$. 
The complexity is much lower than that of the M-FOCUSS algorithm, which is of order $O((QN_rN_p)^3)$. 
However, the complexity of Algorithm 1 is still not satisfied for the 2D-BSCM since $QN_rN_p(N_r+QN_p)$ is still very large. Thus, we need to further reduce the complexity of the proposed algorithm. 

\subsection{Low-complexity 2D-BDCPM acquisition algorithm}
In the previous subsection, we provide a receive model that can be utilized to estimate the
beam domain channel power matrix $\boldsymbol{\Omega}$.  However, the dimensions of $\mathbf{T}_{a}$ and $\mathbf{T}_{f}$ are too large such that the matrix product  $\mathbf{T}_{a}\boldsymbol{\Omega} \mathbf{T}_{f}$ will cause high computational complexity. To reduce the computational complexity of Algorithm \ref{alg:KL},  the structure of $\mathbf{T}_{a}$ and $\mathbf{T}_{f}$ provided in the following theorem and corollary can be utilized.


\begin{theorem}\label{th:theorem of Toeplitz}
	Let the matrix $\bA=\bI_{M,N}\bF_{N}$ be an oversampled DFT matrix and $\bD=\diag{\bd}$ be an $M$-dimensional diagonal matrix. Then, $(\bA^H\bD\bA)\odot(\bA^H\bD\bA)^*$ is a circulant matrix, given by
	\begin{align}
	(\bA^H\bD\bA)\odot(\bA^H\bD\bA)^*=\bF_N^H\bLambda\bF_N,
	\end{align}
	where $\bLambda$ is diagonal matrix defined as
	\begin{align}
	&\bLambda = \frac{1}{N}\diag{\bF_N \left((\bF_N^H\tilde{\bd})\odot(\bF_N^H\tilde{\bd})^*\right) },\\
	&\tilde{\bd}=[\bd^T~\bzero_{N-M,1}^T]^T .
	\end{align}
\end{theorem}
\begin{proof}
	The proof is provided in Appendix \ref{appendice: theorem of Toeplitz}
\end{proof}
Theorem \ref{th:theorem of Toeplitz} is obtained based on the properties of the circulant matrices\cite{circulantmatrices}.
From Theorem \ref{th:theorem of Toeplitz}, we then obtain the following corollary.
\begin{corollary} \label{corol:structure of Ta and Tf}
The matrices $\mathbf{T}_{a}$ and $\mathbf{T}_{f}$ can be written as
\begin{align}
       \bT_a &= (\mathbf{F}_{N_z}\otimes \mathbf{F}_{N_x})^H(\mathbf{\Lambda}_z\otimes\mathbf{\Lambda}_x)(\mathbf{F}_{N_z}\otimes \mathbf{F}_{N_x}) ,
      \\
      \bT_f&= (\bI_{Q}\otimes\mathbf{F}_{N_p})^H\mathbf{\Sigma}(\bI_{Q}\otimes\mathbf{F}_{N_p}),
\end{align}
where $\mathbf{\Lambda}_v$ and $\bLambda_h$ are diagonal matrices, defined as 
\begin{align} 
&\bLambda_v = \frac{1}{{N_z}}\diag{\bF_{N_z}\left( (\bF_{N_z}^H\bd_z)\odot(\bF_{N_z}^H\bd_z)^* \right)},\\
&\bLambda_x = \frac{1}{{N_x}}\diag{\bF_{N_x}\left( (\bF_{N_x}^H\bd_x)\odot(\bF_{N_x}^H\bd_x)^* \right)},\\
& \bd_z = \left[ \bone_{M_z,1}^T~ \bzero_{N_z-M_z,1}^T \right]^T ,\\
& \bd_x = \left[ \bone_{M_x,1}^T~ \bzero_{N_x-M_x,1}^T \right]^T ,
\end{align} 
and $\bSigma$ is a block matrix with diagonal matrices being its elements, defined as 
\begin{align}
&{\mathbf{\Sigma}}
= \begin{bmatrix}
\bSigma_{1,1} & \cdots & \bSigma_{1,Q} \\
\vdots & \ddots & \vdots, \\
\bSigma_{Q,1} & \cdots & \bSigma_{Q,Q} 
\end{bmatrix}, \\
&\bSigma_{q_1,q_2} = \frac{1}{{N_p}}\diag{\bF_{N_p}\left( (\bF_{N_p}^H\bd_{q_1,q_2})\odot(\bF_{N_p}^H\bd_{q_1,q_2})^* \right)}, \\
&\bd_{q_1,q_2} = \left[ \left(\tilde{\bx}_{q_1}\odot\tilde{\bx}_{q_2}^*\right)^T~ \bzero_{N_p-M_p,1}^T \right]^T .
\end{align} 

\end{corollary}
\begin{proof}
	The proof is provided in Appendix \ref{appendix:structure of Ta and Tf}.
\end{proof}

Based on the structures of $\mathbf{T}_{a}$ and $\mathbf{T}_{f}$ provided in Corollary 1, we can reduce the complexity of the proposed algorithm.
In the following, we analyze the
complexity of the proposed algorithm after utilizing the structure provided in Corollary 1.

We write the matrix $\bQ$ as $\left[ \bQ_1~\bQ_2~\cdots~\bQ_Q \right]$ for convenience, where $\bQ_q\in \mathbb{C}^{N_r\times N_p},\forall q$.
According to Corollary \ref{corol:structure of Ta and Tf}, the matrix product  $\bT_a\bQ\bT_f$ can be written as
\begin{align}\label{eq:TaQTf}
\bT_a\bQ\bT_f &= (\mathbf{F}_{N_v}\otimes \mathbf{F}_{N_h})^H(\mathbf{\Lambda}_v\otimes\mathbf{\Lambda}_h)(\mathbf{F}_{N_v}\otimes \mathbf{F}_{N_h}) \bQ
(\bI_{Q}\otimes\mathbf{F}_{N_p})^H\mathbf{\Sigma}(\bI_{Q}\otimes\mathbf{F}_{N_p})\notag \\
&= \begin{bmatrix} 
\sum_{q=1}^{Q}(\mathbf{F}_{N_v}\otimes \mathbf{F}_{N_h})^H(\mathbf{\Lambda}_v\otimes\mathbf{\Lambda}_h)(\mathbf{F}_{N_v}\otimes \mathbf{F}_{N_h}) \bQ_q
\mathbf{F}_{N_p}^H\mathbf{\Sigma}_{q,1}\mathbf{F}_{N_p} \\
\vdots \\
\sum_{q=1}^{Q}(\mathbf{F}_{N_v}\otimes \mathbf{F}_{N_h})^H(\mathbf{\Lambda}_v\otimes\mathbf{\Lambda}_h)(\mathbf{F}_{N_v}\otimes \mathbf{F}_{N_h}) \bQ_q
\mathbf{F}_{N_p}^H\mathbf{\Sigma}_{q,Q}\mathbf{F}_{N_p}
\end{bmatrix}.
\end{align}
There are three kinds of matrix products in (\ref{eq:TaQTf}). For convenience, we analyze their complexities after utilizing the structure as follows.

\begin{enumerate}
\item The first kind of product is the product between an $M\times N$ matrix and an $N$-dimensional
diagonal matrix. Its complexity is  $O(NM)$.  

\item 
The second kind of product is the product between an $M\times N$ matrix and an $N$-dimensional
DFT matrix, which can be implemented by using the fast Fourier transform (FFT) and has complexity of $O(NM\log_2 N)$.
 
\item 

 The third kind product is the product between a matrix and the Kronecker product of
two DFT matrices. For example, we consider the matrix product  $(\bF_{N_1}\otimes \bF_{N_2})\bA$, where 
\begin{equation}
\bA=[\mathrm{vec}(\bA_1)\ \mathrm{vec}(\bA_2)\ \cdots\ \mathrm{vec}(\bA_M)]\in \mathbb{C}^{N\times M} \nonumber 
\end{equation}
 and $\bA_n\in \mathbb{C}^{N_2\times N_1},\forall n$. It can be calculated in the following way as
\begin{align*}
(\bF_{N_1}\otimes \bF_{N2})\bA &= \left[ \mathrm{vec}(\bF_{N2}\bA_1\bF_{N1}^T) ~~ \mathrm{vec}(\bF_{N2}\bA_2\bF_{N1}^T) ~~ \cdots ~~ \mathrm{vec}(\bF_{N2}\bA_M\bF_{N1}^T) \right].
\end{align*}
Its complexity after using FFT is $O(NM\log_2 N)$.
\end{enumerate}

By utilizing the structure provided in Corollary 1, we can obtain a low complexity version of Algorithm 1, the details of which is omitted for brevity.
Base on the complexity analysis of all three kinds of products provided on the above, the computational complexity of
the low-complexity method after utilizing the structure is calculated as  $\mathcal{O}(QN_rN_p\log_2(N_rN_p))$, which is much lower than that of using the direct matrix product.

\subsection{Angle domain BDCPM acquisition method in frequency-flat fading channels}
In the previous subsections, we have introduced the method of obtaining 2D-BDCPM based
on the model in (\ref{eq:receive model}). In fact, the proposed method is applicable as long as the received signal  model satisfies the following form
\begin{align}
\bY=\bA\bG\bB + \bZ,
\end{align}
where $\bA$ and $\bB$ are deterministic matrices, $\bG$ is a random matrix with independent entries,  $\bY$ is the receive matrix, and $\bZ$ is a noise matrix. The matrix $\bOmega=\mathbb{E}\left\{ \bG\odot\bG^* \right\}$ is the statistical parameter to be estimated. 

In this subsection, we present another application of the proposed method. We consider massive
MIMO transmission over frequency-flat fading channels which can be seen as a narrow-band sub-carrier  of the considered massive MIMO-OFDM system. In this system, $K$ UTs equipped with $M_t$ antennas sending pilot signals to a base station equipped with $M_r$ antennas. Let $\mathbf{X}_k\in \mathbb{C}^{M_t\times T}$ denote the uplink pilot signal transmitted by the $k$-th user. We assume the pilot signals of different users are orthogonal to each other for simplicity.

The received channel matrix $\mathbf{H}_{k,m} \in \mathbb{C}^{M_r \times M_t}$ on the $m$-th slot at the BS can be written as
\begin{equation}
         \mathbf{H}_{k,m}  =  \mathbf{V}_r{\mathbf{G}}_{k,m}\bV_t^T,
\end{equation}
where ${\mathbf{G}}_{k,m}$ denotes the angle domain channel matrix, $\mathbf{V}_r$ and $\mathbf{V}_t$ are the conversion matrices from
the angle domain to the space domain at the transmitter and receiver, respectively. Similarly to
the 2D-BSCM, $\mathbf{V}_r$ and $\mathbf{V}_t$ need not be unitary matrices.  

The received pilot signal $\mathbf{Y}_{m} \in \mathbb{C}^{M_r \times T}$  on the $m$-th slot at  the BS  can be written as
\begin{equation}
       \mathbf{Y}_m = \sum\limits_{k=1}^K \mathbf{H}_{k,m}\mathbf{X}_k + \mathbf{Z}_m = \sum\limits_{k=1}^K  \mathbf{V}_r{\mathbf{G}}_{k,m}\bV_t^T\mathbf{X}_k + \mathbf{Z}_m,
\end{equation}
where  
$\mathbf{Z}_m$ is the noise matrix whose elements are independent and identically distributed (i.i.d.) complex Gaussian random variables with zero mean and variance $\sigma_z^2$. 

Under the assumption that the pilot matrices of different UTs are orthogonal to each other, it is easy to obtain  that
\begin{align}
\bV_r^H\mathbf{Y}_m\bX_k^H\bV_t^* = \bV_r^H\bV_r\bG_{k,m}\bV_t^T\bX_k\bX_k^H\bV^*_t + \bV_r^H\bZ_m\bX_k^H\bV_t^*,
\end{align}
By defining $\bT_r$, $\bT_t$ and $\bN$ as
\begin{align}
\bT_r&=(\bV_r^H\bV_r)\odot(\bV_r^H\bV_r)^*, \\
\bT_t&=(\bV_t^T\bX_k\bX_k^H\bV_t^*)\odot(\bV_t^T\bX_k\bX_k^H\bV_t^*)^*, \\
\bN&=\sigma_z^2(\bV_r^T\odot\bV_r^H)\bone((\bX_k^T\bV_t)\odot(\bX_k^H\bV_t^*)),
\end{align}
we can obtain from Theorem \ref{th:theorem_channel_power_matrices_relation} that
\begin{align}
\Expb{(\bV_r^H\mathbf{Y}_m\bX_k\bV_t^*)\odot(\bV_r^H\mathbf{Y}_m\bX_k\bV_t^*)} = \bT_r\bOmega_k \bT_t + \bN,
\end{align}
where the matrix $\bOmega_{k}$ is the angle domain channel power matrix of the $k$-th user, which is defined as $\mathbb{E}\left\{ {\bG_{k,m}\odot\bG_{k,m}^*} \right\}$. It is obviously that the receive model has the same structure as (\ref{eq:statistical_CSI_equation}). Therefore, the proposed method can also be used here to estimate  the angle domain channel power matrix.

\section{Simulation Results}
In this section, we provide simulation results to show the performance of the proposed
algorithm. We adopt both the widely used QuaDRiGa   \cite{quadriga} and the 2D BSCM in \eqref{eq:channel_matrix_correlation_model} to generate channels  for simulation. In the simulation, a massive MIMO system with $M_{r,z}\times M_{r,x}=8\times 16$ is used in most simulations, and an extra-large scale massive MIMO system with $M_{r,z}\times M_{r,x}=16\times 64$ is also considered in the simulations of channel estimation. 
The antenna spacing on the base station side is set
to half wavelength and all UTs are equipped with single antennas. In QuaDRiGa, the scenario is set to “3GPP$\_$3D$\_$UMa$\_$NLOS”. The shadow fading and path
loss are not considered and the UTs in the cell are random uniformly distributed.

The major parameters of OFDM system are summarized in Table \ref{tab:OFDM}. According to these parameters, we can obtain $M_f=\lceil\frac{M_pM_g}{M_c}\rceil=9$. Therefore, the maximum number  of UTs that can be scheduled by orthogonal pilots in an OFDM symbol is $\lfloor \frac{M_p}{M_f} \rfloor = 13$. In the simulations, we select two scenarios with different numbers of UTs, one of which is $K=Q\times P=1\times 12$ and the other is $K=Q\times P=2\times 12$, where $Q$ and $P$ are the number of roots and the number of UTs per root, respectively.  We set the fine factors as $N_{a,z}=N_{a,x}=N_{a,f}=2$, which is enough to obtain good performance. Furthermore, the power of the transmitted pilots is set to 1. The signal-to-noise ratio (SNR) is given by $\mathrm{SNR}=\frac{1}{\sigma_z^2}$.
\begin{table}[H]
	\centering
	\caption{OFDM parameters}
	\label{tab:OFDM}
	\begin{tabular}{ll}
		\hline
		Parameters & Value \\
		\hline
		Carrier frequency $f_c$ &  4.8 GHz \\
		Subcarrier spacing $\Delta f$ & 30 kHz \\
		Guard interval $M_g$ & 144 \\
		Subcarriers number $M_c$ & 2048 \\
		Transmitting subcarriers number $M_p$ & 120 \\
		\hline
	\end{tabular}
\end{table}

\begin{figure}[h]
	\centering
	\includegraphics[width=0.6\linewidth]{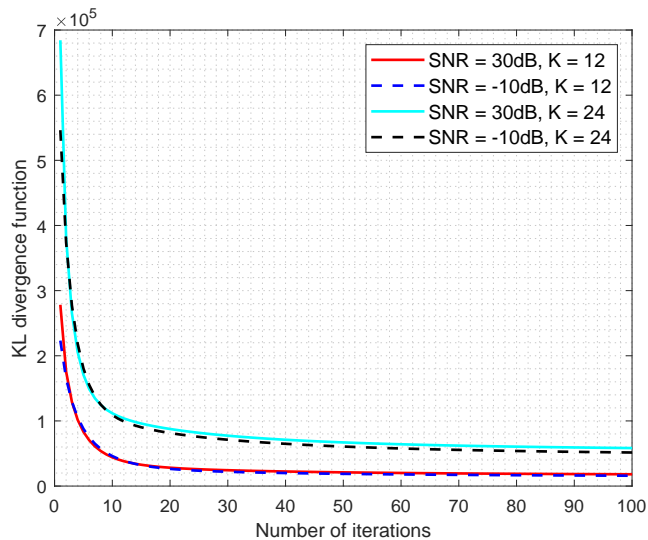}
	\caption{The objective function value under different iteration times in massive MIMO system with $M_{r,z}=8,M_{r,x}=16$, where the channel is generated from QuaDRiGa.}
	\label{fig:convergence performance}
\end{figure}

First, we use massive MIMO channels generated from QuaDRiGa for simulation. The aim is mainly to verify the convergence of the proposed algorithm and analyze the influence of noise and pilot interference on
the estimated BDCPM through the angle domain power spectrum, i.e., the mesh graph of the
BDCPM. In Fig.~\ref{fig:convergence performance}, the convergence performance of the proposed algorithm in four scenarios in massive MIMO system with $M_{r,z}=8,M_{r,x}=16$
are showed. The SNRs under consideration are $-10$dB and $30$dB, and the numbers of UTs are
$12$ and $24$. It can be observed that, the proposed algorithm can approach convergence within $20$
iterations in all the scenarios. Furthermore, the algorithm converges faster with fewer UTs than with more UTs.

\begin{figure}[htbp]
	\centering
	\subfigure[]{
		\begin{minipage}[t]{0.4\linewidth}
			\centering
			\includegraphics[width=2.5in]{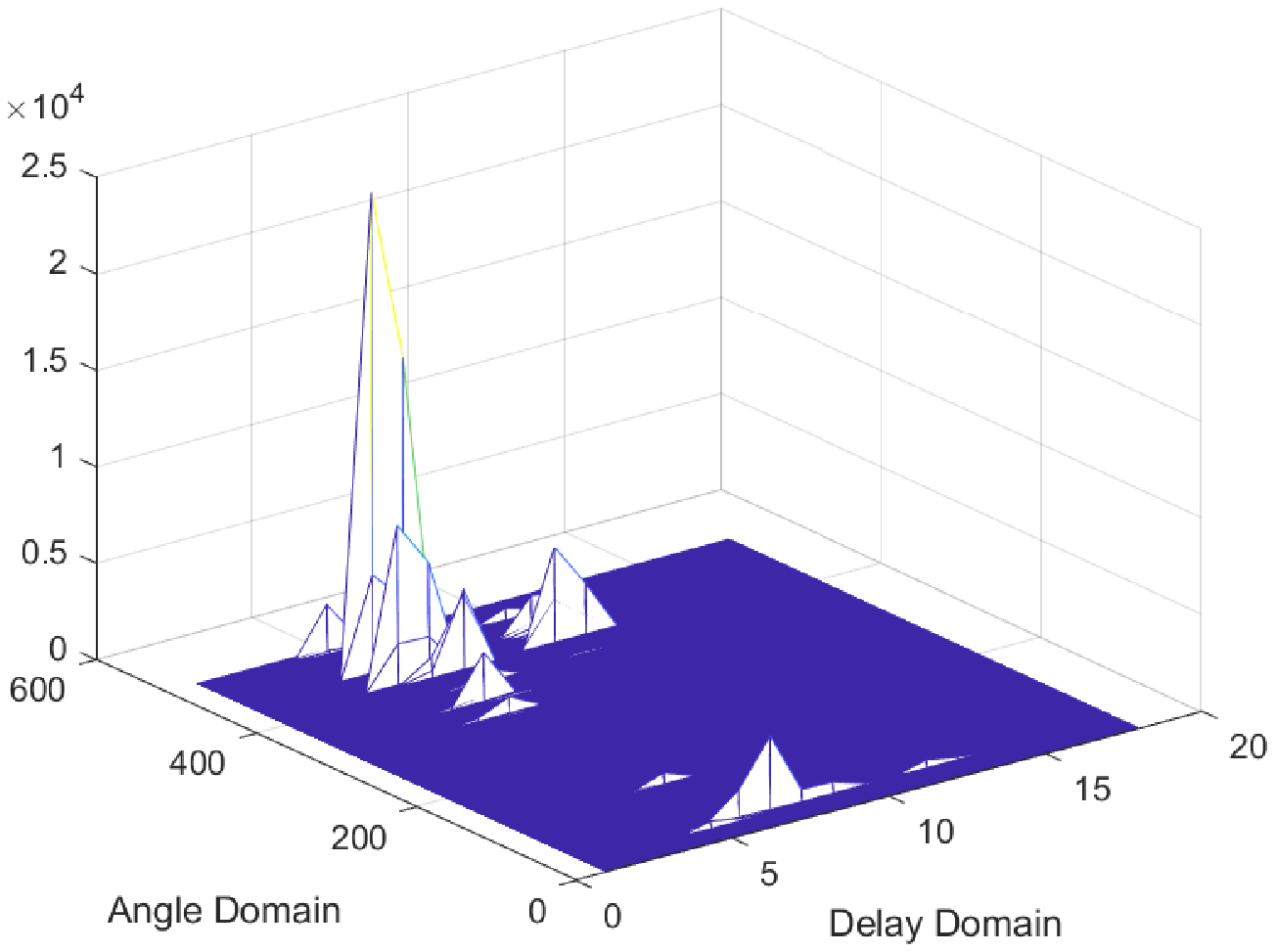}
		\end{minipage}%
	}%
	\subfigure[]{
		\begin{minipage}[t]{0.4\linewidth}
			\centering
			\includegraphics[width=2.5in]{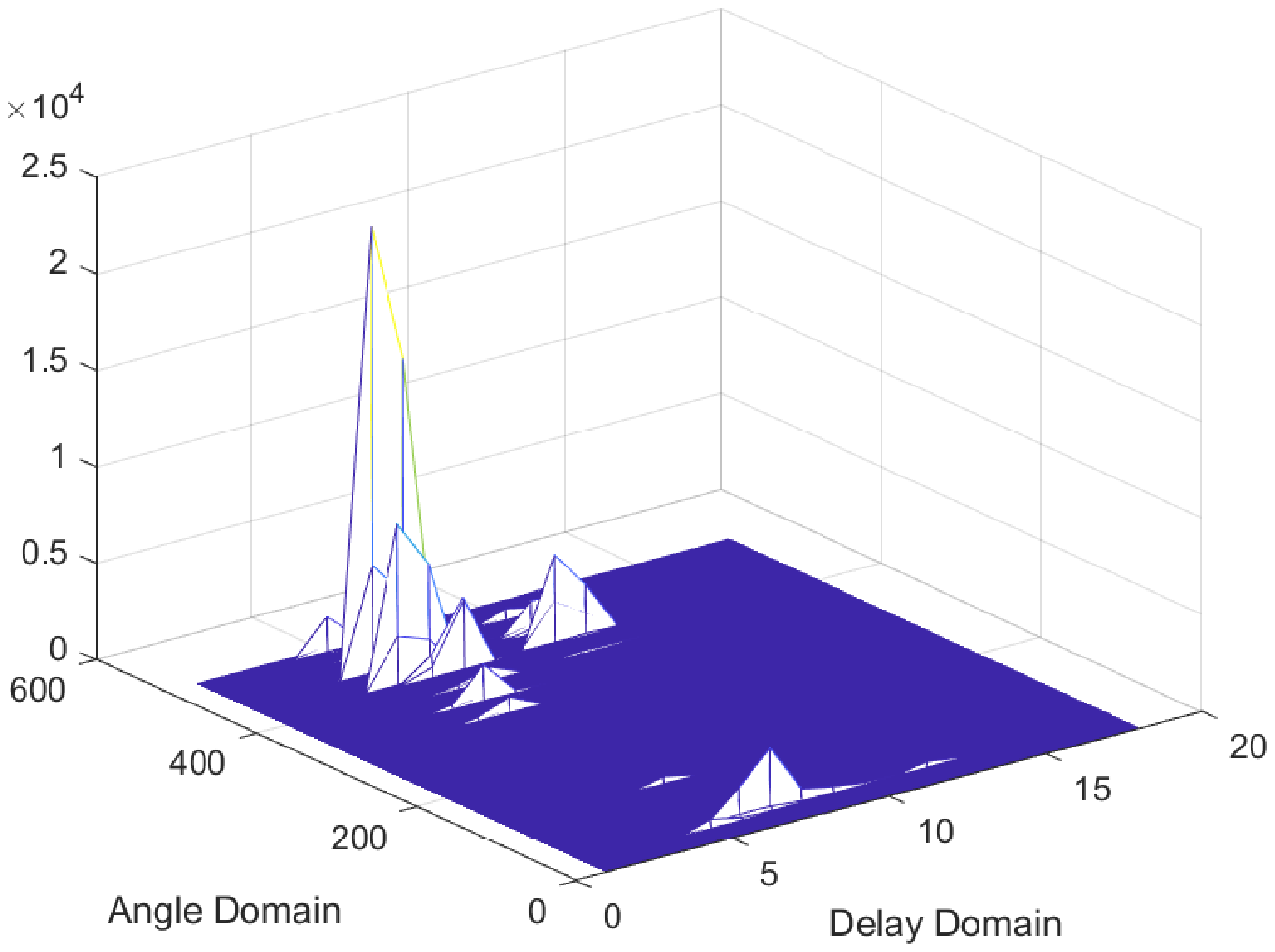}
		\end{minipage}%
	}%

	\subfigure[]{
		\begin{minipage}[t]{0.4\linewidth}
			\centering
			\includegraphics[width=2.5in]{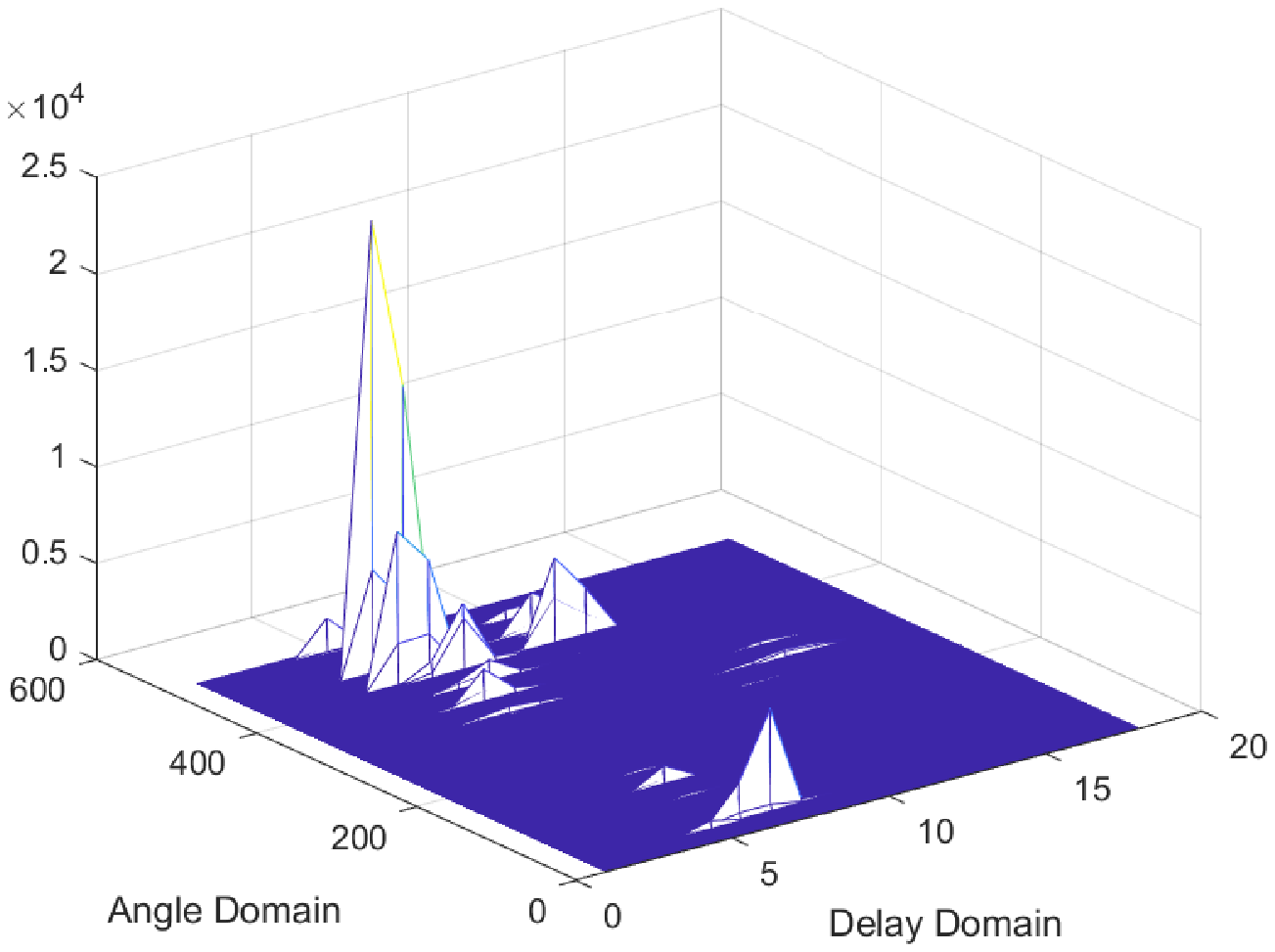}
		\end{minipage}
	}%
	\subfigure[]{
		\begin{minipage}[t]{0.4\linewidth}
			\centering
			\includegraphics[width=2.5in]{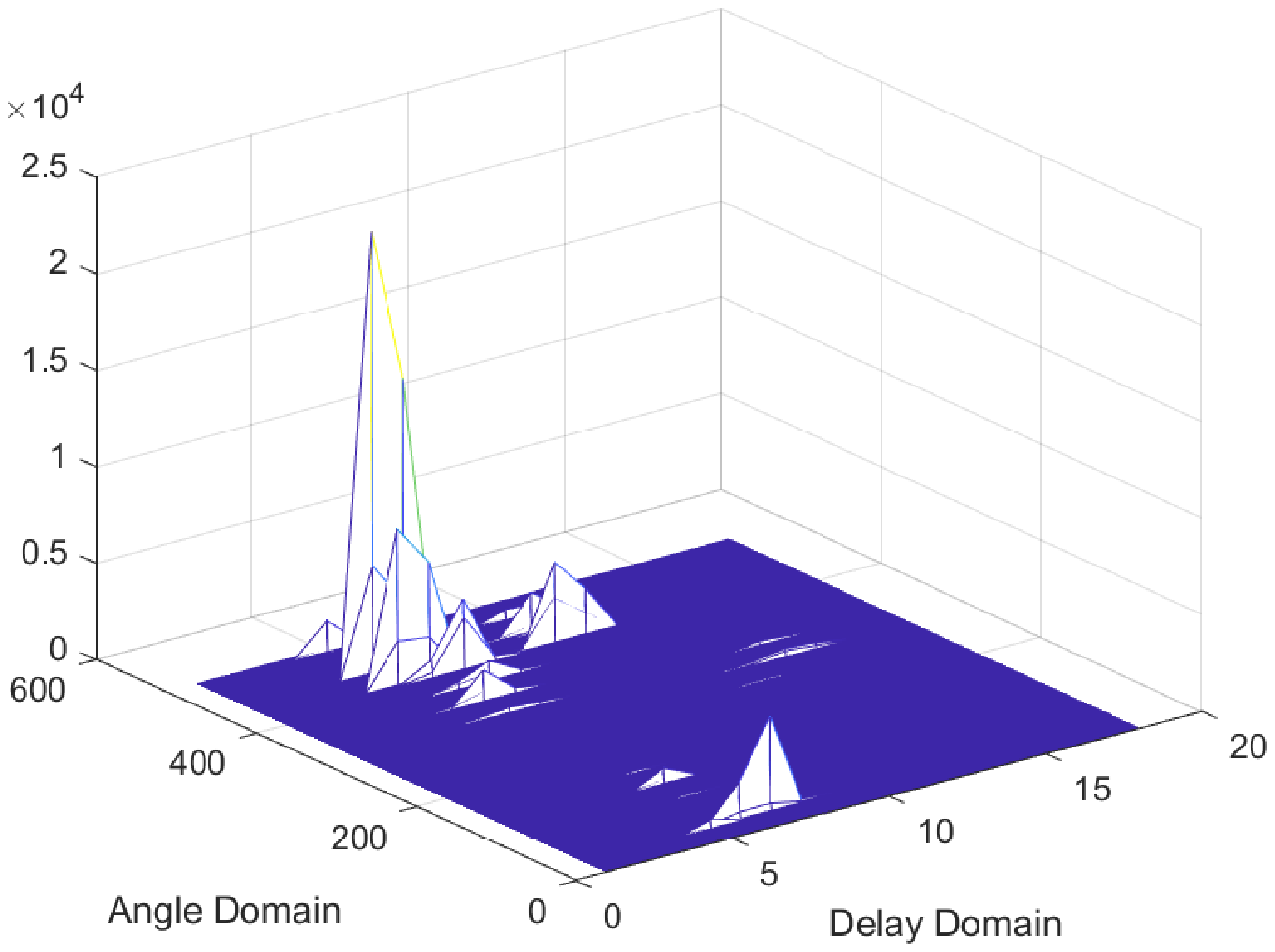}
		\end{minipage}
	}%
	\centering
	\caption{Estimated angle delay domain power spectrum of the first user in a massive MIMO system with $M_{r,z}=8, M_{r,x}=16$ and (a)SNR = $30$dB, $K = 12$, (b)SNR=$-10$dB,  $K = 12$,  (c)SNR = $30$dB, $K = 24$, (d)SNR = $-10$dB, $K = 24$. The channel is generated from QuaDRiGa.} 
	\label{fig:mesh graph}
\end{figure}
 
In Fig.~\ref{fig:mesh graph}, we give the estimated angle delay domain power spectrum of the first user under different SNRs and numbers of UTs for the considered massive MIMO. The first user is the same user for $K=12$ and $K=24$. By comparing Fig.~\ref{fig:mesh graph}(a) with Fig.~\ref{fig:mesh graph}(b) or Fig.~\ref{fig:mesh graph}(c) with Fig.~\ref{fig:mesh graph}(d), it can be found that reducing the SNR has little effect on the obtained BDCPM. This indicates that the proposed algorithm has a prominent anti-noise effect. Similarly, by comparing Fig.~\ref{fig:mesh graph}(a) with Fig.~\ref{fig:mesh graph}(c) or Fig.~\ref{fig:mesh graph}(b) with Fig.~\ref{fig:mesh graph}(d), we can find that the BDCPM in Fig.~\ref{fig:mesh graph}(c) or Fig.~\ref{fig:mesh graph}(d) has more non-zero beams, which is mainly caused by the interference of the pilot signals on another root. Therefore, the pilot interference on other roots will reduce the accuracy of the BDCPM to some extent.

\begin{figure} 
	\centering
	\includegraphics[width=0.6\linewidth]{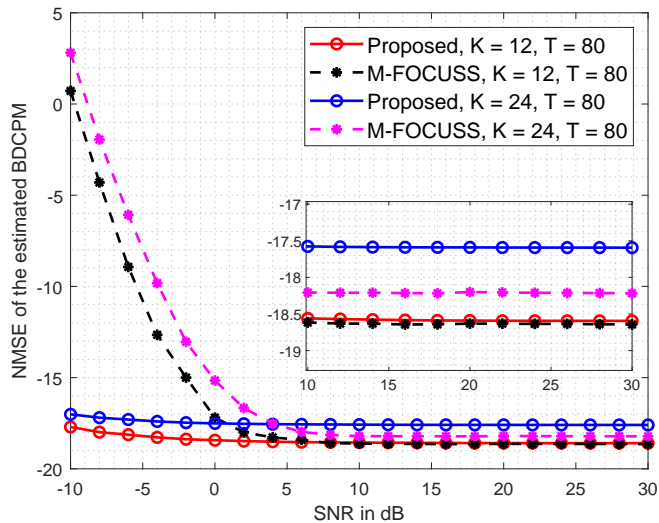}
	\caption{NMSE of the estimated BDCPM  versus SNRs for four  scenarios in a massive MIMO system with $M_{r,z}=8,M_{r,x}=16, T=80$, where the channel is generated from the 2D BSCM, the number of UTs is 12 and 24, and the methods are the M-FOCUSS algorithm and the proposed algorithm. }
	\label{fig:NMSE of Massive MIMO}
\end{figure}

Then, we evaluate the accuracy of the estimated BDCPM using the  channels generated from the 2D BSCM  with given BDCPM. The results of the proposed algorithm are used to compare with that of the M-FOCUSS  algorithm, which has been verified to have high accuracy under high SNR among a series of compressed sensing algorithms. Due to the high implementation complexity of M-FOCUSS in extra-large scale massive MIMO, we only compare the accuracy performance in massive MIMO simulation. Since the accurate BDCPM is known, we can use the normalized mean squared error (NMSE) between the estimated BDCPM $\hat{\bOmega}_k$ and the accurate BDCPM ${\bOmega}_k$ to evaluate the accuracy of the BDCPM. The NMSE in dB is defined as
\begin{equation}
	\mathrm{NMSE(dB)}\triangleq 10\log_2\left( \frac{1}{K}\sum_{k=1}^K \dfrac{\left\Arrowvert \hat{\bOmega}_{k} - {\bOmega}_{k} \right\Arrowvert_F^2}{\left\Arrowvert {\bOmega}_{k} \right\Arrowvert_F^2} \right)
\end{equation}
The simulation results of NMSE performance  of the estimated BDCPM are shown in Fig.~\ref{fig:NMSE of Massive MIMO}, where the massive MIMO system is also with $M_{r,z}=8,M_{r,x}=16, T=80$. It can be observed that the accuracy of the BDCPM obtained by the proposed algorithm is not less than that obtained by the M-FOCUSS method, no matter  orthogonal pilots with $K=12$  or non-orthogonal pilots with $K=24$ are used. Meanwhile, it can also be found that the proposed algorithm has strong anti-noise performance, while the accuracy of the M-FOCUSS algorithm drops sharply when the SNR is less than 0 dB. 

\begin{figure} 
	\centering
	\includegraphics[width=0.6\linewidth]{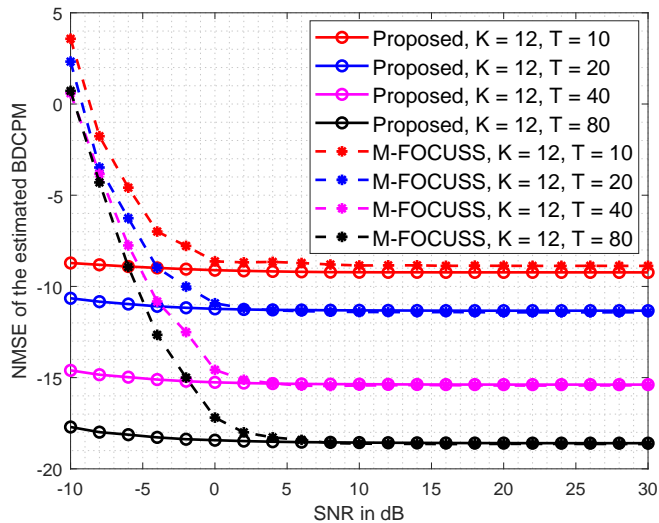}
	\caption{NMSE of the estimated BDCPM versus SNRs for four different numbers of samples in a massive MIMO system with $M_{r,z}=8,M_{r,x}=16, K=12$, where the channel is generated from the 2D BSCM  and the methods are the M-FOCUSS algorithm and the proposed algorithm. }
	\label{fig:NMSE of Massive MIMO Samples K=12}
\end{figure}
\begin{figure} 
	\centering
	\includegraphics[width=0.6\linewidth]{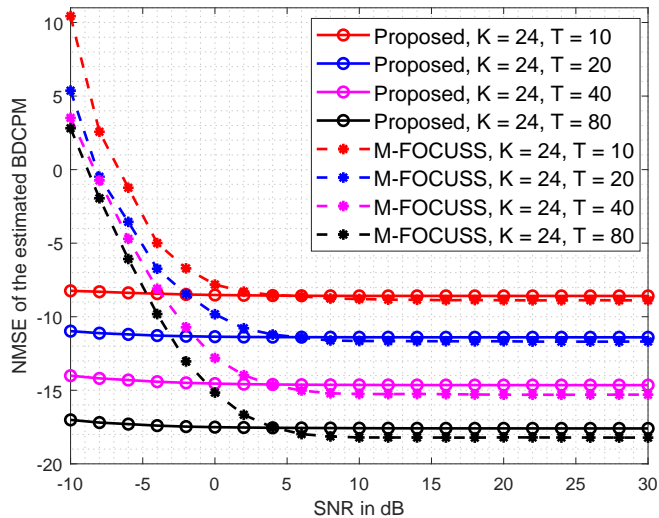}
	\caption{NMSE of the estimated BDCPM versus SNRs for different numbers of samples  in a massive MIMO system with $M_{r,z}=8,M_{r,x}=16, K=24$, where the channel is generated from the 2D BSCM and the methods are the M-FOCUSS  algorithm and the proposed algorithm. }
	\label{fig:NMSE of Massive MIMO Samples K=24}
\end{figure}

The accuracy of the estimated BDCPM also depends on the number of  received pilot signals, which has been set as $T=80$ in the simulations for Fig.~\ref{fig:NMSE of Massive MIMO}.  To show the relation between the NMSE performance of the estimated BDCPM and the number of samples used in the estimation,
we simulate the NMSE of the estimated BDCPM for different numbers of samples in the considered massive MIMO system with $M_{r,z}=8,M_{r,x}=16$. The numbers of samples are set as $T=10, 20, 40, 80$.
The simulation results of $K=12$ and $K=24$ for all cases are provided in  Fig.~\ref{fig:NMSE of Massive MIMO Samples K=12} and  
Fig.~\ref{fig:NMSE of Massive MIMO Samples K=24}, respectively.
We observe that the NMSE of the estimated BDCPM for  both the M-FOCUSS algorithm and the proposed algorithm can achieve close to $-10$dB performance with only 10 samples of receive pilot signals. Furthermore, the NMSE performances of the two algorithms in high SNR regime decreases almost linearly as the number of received pilot signals increases.
Thus, we do not need too many samples of receive pilot signals in practical massive MIMO systems to obtain the statistical channel information with good accuracy. 
Finally, the proposed algorithm greatly outperforms the M-FOCUSS algorithm
	in the low SNR regime, and achieves nearly the same performance as that of the latter in the high SNR regime for all cases.
It indicates that the proposed algorithm is robust to the noise, whereas the M-FOCUSS method need a larger SNR to work.

Next, we use the estimated BDCPM for the estimation of instantaneous CSI. To show the performance of channel estimation in a more realistic scenario, the channels are generated by QuaDRiGa rather than the 2D BSCM. We consider both the massive MIMO and extra-large scale MIMO case. Due to its high complexity, 
the M-FOCUSS method is only used in massive MIMO scenarios. The MMSE algorithm is used for channel estimation.
We define the mean square error of space-frequency domain channel in dB as
\begin{align}
\mathrm{MSE(dB)}\triangleq 10\log_2\left( \frac{1}{KT}\sum_{k=1}^{K}\sum_{t=1}^{T} \left\Arrowvert \hat{\bH}_{k,t} - {\bH}_{k,t} \right\Arrowvert_F^2 \right),
\end{align}
where $ \left\Arrowvert \cdot \right\Arrowvert_F$ represents the Frobenius norm and the space-frequency domain channel ${\bH}_{k,t}$ has been normalized as $\left\Arrowvert{\bH}_{k,t}\right\Arrowvert_F^2=M_rM_p$. 
The simulation results of the channel estimation performance of the considered massive MIMO and extra-large scale MIMO are shown in Fig. \ref{fig:MSE of Massive MIMO} and Fig. \ref{fig:MSE of extra-large scale massive MIMO}, respectively. It can be observed in Fig. \ref{fig:MSE of Massive MIMO} that when the SNR is less than 10 dB, the proposed algorithm can bring significant channel estimation performance gain compared to the M-FOCUSS algorithm. When the SNR is greater than 10 dB, the proposed algorithm can also obtain comparable performance to the M-FOCUSS algorithm. Since the complexity of the proposed algorithm  is much lower than that of the M-FOCUSS algorithm, it is superior to the M-FOCUSS algorithm in estimating the statistical channel information. Finally, the effectiveness of the proposed algorithm at an extra-large scale massive MIMO is also verified in Fig. \ref{fig:MSE of extra-large scale massive MIMO}.

\begin{figure} 
	\centering
	\includegraphics[width=0.6\linewidth]{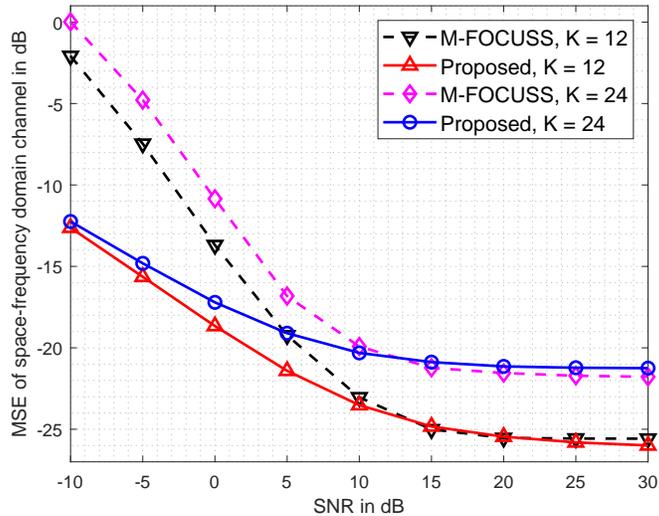}
	\caption{MSE versus SNR in a massive MIMO system with $M_{r,z}=8,M_{r,x}=16$, where the channel is generated from QuaDRiGa. Four scenarios are compared where the number of UTs is 12 and 24 and the methods are M-FOCUSS and the proposed algorithm, respectively. }
	\label{fig:MSE of Massive MIMO}
\end{figure}
\begin{figure} 
	\centering
	\includegraphics[width=0.6\linewidth]{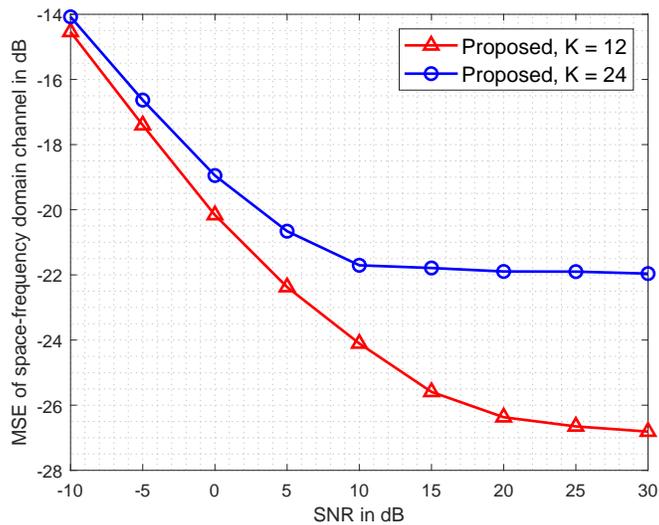}
	\caption{MSE versus SNR in an extra-large scale massive MIMO system with $M_{r,z}=16,M_{r,x}=64$, where the channel is generated from QuaDRiGa. Two scenarios are compared where the number of UTs is 12 and 24, respectively.}
	\label{fig:MSE of extra-large scale massive MIMO}
\end{figure}



\section{Conclusion}
In this paper,  the beam domain statistical CSI  estimation for the  2D-BSCM  in massive MIMO systems was investigated.  We considered the problem of estimating the BDCPMs based on multiple receive pilot signals. From the 2D-BSCM, we derived a receive model, which shows the relation between the statistical property of the receive pilot signals and the BDCPMs. Then, we formulated an optimization problem based on the Kullback-Leibler (KL) divergence and the receive model.
A novel method that estimates the BDCPMs without involving instantaneous CSI was proposed by solving the optimization problem. The proposed method has much lower complexity than that of the M-FOCUSS algorithm. The complexity of the proposed method was then further reduced by utilizing the circulant structures.  
Simulations results showed the proposed method works well and bring significant performance gain when used in channel estimation.
\appendices

\section{Proof of Theorem \ref{th:theorem_channel_power_matrices_relation}}
\label{appendice: theorem_channel_power_matrices_relation}
We define the matrix $\bGamma$ as  $\bGamma=\mathbb{E}\{(\mathbf{C}_1{\mathbf{G}}_{t}\mathbf{C}_2) \odot (\mathbf{C}_1{\mathbf{G}}_{t}\mathbf{C}_2) ^*\}$ for convenience. Then, the
entires  $[\bGamma]_{ij}$ can be specifically represented as
\begin{align}
[\bGamma]_{ij} &= \mathbb{E}\{[\mathbf{C}_1{\mathbf{G}}_{t}\mathbf{C}_2]_{ij} \odot [\mathbf{C}_1{\mathbf{G}}_{t}\mathbf{C}_2]_{ij} ^*\}
\notag \\ 
&= \Expect{}{\left( \sum_{p,q}[\bC_1]_{ip}[\bG_t]_{pq}[\bC]_{qj} \right)\left( \sum_{p',q'}[\bC_1]_{ip'}[\bG_t]_{p'q'}[\bC]_{q'j} \right)^*} \notag \\ 
&= \sum_{p,q,p',q'}[\bC_1]_{ip}[\bC_1]_{ip'}^*\Expect{}{[\bG_t]_{pq}[\bG_t]_{p'q'}^*}[\bC_2]_{qj}[\bC_2]_{q'j}^* .
\end{align}
Based on the assumption that each element of $\bG_t$ has zero mean and is independent of each other, i.e., $\Expect{}{[\bG_t]_{pq}[\bG_t]_{p'q'}^*}=\delta(p-p')\delta(q-q')[\bOmega]_{pq}$, we can simplify the expression of $[\bGamma]_{ij}$ to
\begin{align}
[\bGamma]_{ij} 
&=\sum_{p,q}\left| [\bC_1]_{ip} \right|^2 [\bOmega]_{pq}\left| [\bC_2]_{qj} \right|^2 .
\end{align}
This can be organized in a matrix form as
\begin{align}
\bGamma = (\bC_1\odot\bC_1^*)\bOmega(\bC_2\odot\bC_2^*).
\end{align}
By defining $\bT_1=\bC_1\odot\bC_1^*$ and $\bT_2=\bC_2\odot\bC_2^*$, we then obtain
\begin{align}
\mathbb{E}\{(\mathbf{C}_1{\mathbf{G}}_{t}\mathbf{C}_2) \odot (\mathbf{C}_1{\mathbf{G}}_{t}\mathbf{C}_2) ^*\}=  \mathbf{T}_1\boldsymbol{\Omega}\mathbf{T}_2,
\end{align}
Thus, we conclude the proof.

\section{Proof of Theorem \ref{th:theorem_of_gradients}}
\label{appendix:theorem of gradients}
We calculate the gradient $\frac{\partial  [\mathbf{T}_{a}(\mathbf{M} \odot \mathbf{M})\mathbf{T}_{f}]_{ij}}{\partial {\mathbf{M}}}$ first. The entries $[\mathbf{T}_{a}(\mathbf{M} \odot \mathbf{M})\mathbf{T}_{f}]_{ij}$ can be written as
\begin{align}
 [\mathbf{T}_{a}(\mathbf{M} \odot \mathbf{M})\mathbf{T}_{f}]_{ij} 
=  \sum_{p,q} [\bT_a]_{ip}[\bM\odot\bM]_{pq}[\bT_f]_{qj}  
\end{align}
We define $\mathbf{e}_i$ as the column vector whose only nonzero entry is its $i$-th element, which has value
of $1$. From $[\bT_a]_{ip}[\bT_f]_{qj} = \left[ \bT_a^T\be_i\be_j^T\bT_f^T \right]_{pq}$, we have that
\begin{align}
 [\mathbf{T}_{a}(\mathbf{M} \odot \mathbf{M})\mathbf{T}_{f}]_{ij}  
=& \sum_{p,q} \left[ \bT_a^T\be_i\be_j^T\bT_f^T \right]_{pq} [\bM\odot\bM]_{pq} \notag \\  
=& \sum_{p,q} \left[ \left( \bT_a^T\be_i\be_j^T\bT_f^T \right)\odot (\mathbf{M} \odot \mathbf{M}) \right]_{pq}.
\end{align}
Then, we can obtain
\begin{align}
\frac{\partial  [\mathbf{T}_{a}(\mathbf{M} \odot \mathbf{M})\mathbf{T}_{f}]_{ij}}{\partial {[\mathbf{M}]_{pq}}} = 2\left[ \bT_a^T\be_i\be_j^T\bT_f^T \right]_{pq}[\bM]_{pq},
\end{align}
so the gradient $\frac{\partial  [\mathbf{T}_{a}(\mathbf{M} \odot \mathbf{M})\mathbf{T}_{f}]_{ij}}{\partial {\mathbf{M}}}$ is given by
\begin{align}
\frac{\partial  [\mathbf{T}_{a}(\mathbf{M} \odot \mathbf{M})\mathbf{T}_{f}]_{ij}}{\partial {\mathbf{M}}} = 2\left( \bT_a^T\be_i\be_j^T\bT_f^T \right)\odot \mathbf{M} .
\end{align}

Next, the first target gradient can be calculated as
\begin{align}\label{eq:calc of partial1}
 \frac{\partial\sum_{ij}  [\mathbf{T}_{a}(\mathbf{M} \odot \mathbf{M})\mathbf{T}_{f}]_{ij}}{\partial {\mathbf{M}}} 
=& 2\left( \sum_{ij}\bT_a^T\be_i\be_j^T\bT_f^T \right) \odot \bM \notag \\
=& 2\left( \bT_a^T\left( \sum_{ij}\be_i\be_j^T \right)\bT_f^T \right) \odot \bM \notag \\
=& 2\left( \mathbf{T}_{a}^T\bone\mathbf{T}_{f}^T \right) \odot  \mathbf{M} .
\end{align}
Furthermore, the second target gradient can be calculated as
\begin{align}
&\frac{\partial \sum_{ij}[\mathbf{\Phi}]_{ij}\log \left[\mathbf{T}_{a}(\mathbf{M} \odot \mathbf{M})\mathbf{T}_{f}+ \mathbf{N} \right]_{ij}}{ \partial \mathbf{M}} \notag \\
=& \sum_{ij} \dfrac{[\bPhi]_{ij}}{[\mathbf{T}_{a}(\mathbf{M} \odot \mathbf{M})\mathbf{T}_{f}+ \mathbf{N}]_{ij}}  \dfrac{\partial  [\mathbf{T}_{a}(\mathbf{M} \odot \mathbf{M})\mathbf{T}_{f}]_{ij}}{\partial {\mathbf{M}}} \notag \\
=& 2\sum_{ij} \dfrac{[\bPhi]_{ij}}{[\mathbf{T}_{a}(\mathbf{M} \odot \mathbf{M})\mathbf{T}_{f}+ \mathbf{N}]_{ij}}  \left( \bT_a^T\be_i\be_j^T\bT_f^T \right)\odot \mathbf{M} .
\end{align}
By defining $\bQ$ as $[\mathbf{Q}]_{ij}=\frac{[\mathbf{\Phi}]_{ij}}{[\mathbf{T}_{a}\boldsymbol{\Omega}\mathbf{T}_{f}+  \mathbf{N}]_{ij}}$, we then obtain
\begin{align}\label{eq:calc of partial2}
&\frac{\partial \sum_{ij}[\mathbf{\Phi}]_{ij}\log \left[\mathbf{T}_{a}(\mathbf{M} \odot \mathbf{M})\mathbf{T}_{f}+ \mathbf{N} \right]_{ij}}{ \partial \mathbf{M}} \notag \\
=& 2\left( \sum_{ij}[\bQ]_{ij}\mathbf{T}_{a}^T\be_i  \be_j^T\mathbf{T}_{f}^T \right) \odot \bM \notag \\
=& 2\left( \mathbf{T}_{a}^T\left( \sum_{ij}\be_i[\bQ]_{ij}\be_j^T  \right) \mathbf{T}_{f}^T \right) \odot \bM \notag \\
=& 2\left(\mathbf{T}_{a}^T\mathbf{Q}\mathbf{T}_{f}^T\right) \odot \mathbf{M}.
\end{align} 
Since $\bT_a$ and $\bT_f$ are symmetric matrices, then (\ref{eq:calc of partial1}) and (\ref{eq:calc of partial2}) can also be written as
	\begin{align}
	&\frac{\partial\sum_{ij}  [\mathbf{T}_{a}(\mathbf{M} \odot \mathbf{M})\mathbf{T}_{f}]_{ij}}{\partial {\mathbf{M}}}
	= 2\left( \mathbf{T}_{a}\bone\mathbf{T}_{f} \right) \odot  \mathbf{M},
	\\
	& \frac{\partial \sum_{ij}[\mathbf{\Phi}]_{ij}\log \left[\mathbf{T}_{a}(\mathbf{M} \odot \mathbf{M})\mathbf{T}_{f}+ \mathbf{N} \right]_{ij}}{ \partial \mathbf{M}}
	=  2\left(\mathbf{T}_{a}\mathbf{Q}\mathbf{T}_{f}\right) \odot \mathbf{M}.
	\end{align}

\section{Proof of Theorem \ref{th:theorem of Toeplitz}}
\label{appendice: theorem of Toeplitz}
From the properties of the circulant matrix \cite{circulantmatrices}, we can know that $\bA^H\bD\bA$ is a circulant matrix when $\bD=\diag{\bd}$ is a diagonal matrix and $\bA=\bI_{M,N}\bF_N$ is an oversampled DFT matrix. Therefore, $\bA^H\bD\bA$ can be represented in the following form
\begin{align}
\bA^H\bD\bA = \sum_{i=0}^{N-1}[\bc]_i\bPi_N^{N-i},
\end{align}
where $\bc$ is the first column of the matrix $\bA^H\bD\bA$, which can be calculated as
\begin{align}
\bc &= (\bA^H\bD\bA)\be_1 \notag \\
&=\bA^H\bd \notag \\
& = \bF_N^H\tilde{\bd}.
\end{align}
The vector $\tilde{\bd}$ is defined as $[{\bd}^T~\bzero_{N-M,1}^T]^T$. Then, we have
\begin{align}
(\bA^H\bD\bA)\odot(\bA^H\bD\bA)^* &= \left(\sum_{i=0}^{N-1} [\bc]_i\bPi_N^{N-i}\right) \odot \left(\sum_{j=0}^{N-1} [\bc]_j\bPi_N^{N-j}\right)^* \notag \\
&= \sum_{i=0}^{N-1}\sum_{j=0}^{N-1}[\bc]_i[\bc]_j^*\left( \bPi_N^{N-i}\odot\bPi_N^{N-j} \right)\notag \\
&\overset{(a)}{=} \sum_{i=0}^{N-1}[\bc]_i[\bc]_i^*\bPi_N^{N-i} .
\end{align}
Equation (a) holds because $\bPi_N^{i}\odot\bPi_N^{j} = \delta(i-j)\bPi_N^{i}$. Obviously, $(\bA^H\bD\bA)\odot(\bA^H\bD\bA)^*$ is also a circulant matrix with $\bc\odot\bc^*$ in the first column. Thus, it can be represented as
\begin{align}
(\bA^H\bD\bA)\odot(\bA^H\bD\bA)^* = \bF_N^H\bLambda\bF_{N},
\end{align}
where $\bLambda=\diag{\blambda}$ is a diagonal matrix and satisfies
\begin{align}
\bc\odot\bc^*=\bF_N^H\blambda.
\end{align}
Then, we obtain
\begin{align}
\bLambda&=\frac{1}{N}\diag{\bF_N (\bc\odot\bc^*) }\notag \\
&= \frac{1}{N}\diag{\bF_N \left((\bF_N^H\tilde{\bd})\odot(\bF_N^H\tilde{\bd})^*\right) }.
\end{align}
Thus, we conclude the proof.

\section{Proof of Corollary \ref{corol:structure of Ta and Tf}}
\label{appendix:structure of Ta and Tf}
Let us review the formulas of $\bT_a$ and $\bT_f$ from equation (\ref{eq:Ta}) and (\ref{eq:Tf}). First, we substitute $\bV=\bV_z\otimes\bV_x$ into (\ref{eq:Ta}), as following
\begin{align}
\bT_a &= (\bV^H\bV)\odot(\bV^H\bV)^* \notag \\
&= \left((\bV_z\otimes\bV_x)^H(\bV_z\otimes\bV_x)\right)\odot\left((\bV_z\otimes\bV_x)^H(\bV_z\otimes\bV_x)\right)^* \notag \\
&=\left((\bV_z^H\bV_z)\otimes(\bV_x^H\bV_x)\right)\odot\left((\bV_z^H\bV_z)\otimes(\bV_x^H\bV_x)\right)^* \notag \\
&=\left((\bV_z^H\bV_z)\odot(\bV_z^H\bV_z)^*\right)\otimes\left((\bV_x^H\bV_x)\odot(\bV_x^H\bV_x)^*\right).
\end{align}
From Theorem \ref{th:theorem of Toeplitz}, and the fact that $\bV_z=\bI_{M_z,N_z}\bF_{N_z}$ and $\bV_x=\bI_{M_x,N_x}\bF_{N_x}$ are oversampled DFT matrices, we obtain that $(\bV_z^H\bV_z)\odot(\bV_z^H\bV_z)^*$ and $(\bV_x^H\bV_x)\odot(\bV_x^H\bV_x)^*$ are circulant matrices, written as 
\begin{align}
&(\bV_z^H\bV_z)\odot(\bV_z^H\bV_z)^* = \bF_{N_{z}}^H\bLambda_{z}\bF_{N_{z}}, \\ 
&(\bV_x^H\bV_x)\odot(\bV_x^H\bV_x)^* = \bF_{N_{x}}^H\bLambda_{x}\bF_{N_{x}},
\end{align}	 
where $\mathbf{\Lambda}_v$ and $\bLambda_h$ are diagonal matrices, defined as 
	\begin{align} 
	&\bLambda_z = \frac{1}{N_z}\diag{\bF_{N_z}\left( (\bF_{N_z}^H\bd_z)\odot(\bF_{N_z}^H\bd_z)^* \right)},\\
	&\bLambda_x = \frac{1}{N_x}\diag{\bF_{N_x}\left( (\bF_{N_x}^H\bd_x)\odot(\bF_{N_x}^H\bd_x)^* \right)},\\
	& \bd_z = \left[ \bone_{M_z,1}^T~ \bzero_{N_z-M_z,1}^T \right]^T, \\
	& \bd_x = \left[ \bone_{M_x,1}^T~ \bzero_{N_x-M_x,1}^T \right]^T. 
	\end{align} 
Therefore, $\bT_a$ can be represented as
\begin{align}\label{eq:Ta old}
\bT_a &= \left( \bF_{N_{z}}^H\bLambda_{z}\bF_{N_{z}} \right)\otimes\left( \bF_{N_{x}}^H\bLambda_{x}\bF_{N_{x}} \right)\notag \\
&= (\mathbf{F}_{N_z}\otimes \mathbf{F}_{N_x})^H(\mathbf{\Lambda}_z\otimes\mathbf{\Lambda}_x)(\mathbf{F}_{N_z}\otimes \mathbf{F}_{N_x}) .
\end{align}

Next, we derive the DFT structure of $\bT_f$. Because $\bT_f$ in (\ref{eq:Tf}) consists of a series of sub-matrices, i.e.,
\begin{align*}
\left(\mathbf{U}^T\tilde{\bX}_{q_1}\tilde{\bX}_{q_2}^H\mathbf{U}^*\right)\odot \left(\mathbf{U}^T\tilde{\bX}_{q_1}\tilde{\bX}_{q_2}^H\mathbf{U}^*\right)^*,\ q1,q2=1,..,Q.
\end{align*}
Since $\bU=\bI_{M_p,N_p}\bF_{N_p}$ is an oversampled DFT matrix and $\tilde{\bX}_{q_1}\tilde{\bX}_{q_2}^H=\diag{\tilde{\bx}_{q_1}\odot\tilde{\bx}_{q_2}^*}$ is a diagonal matrix, we obtain that each submatrix can be represented as
\begin{align}\label{eq:Tf submatrix}
\left(\mathbf{U}^T\tilde{\bX}_{q_1}\tilde{\bX}_{q_2}^H\mathbf{U}^*\right)\odot \left(\mathbf{U}^T\tilde{\bX}_{q_1}\tilde{\bX}_{q_2}^H\mathbf{U}^*\right)^*
= \bF_{N_p}^H \bSigma_{q_1,q_2} \bF_{N_p},
\end{align}
where $\bSigma_{q_1,q_2}$ is a diagonal matrix, defined as 
\begin{align}
\label{eq:Sigma q1 q2}\bSigma_{q_1,q_2} &= \frac{1}{N_p}\diag{\bF_{N_p}\left( (\bF_{N_p}^H\bd_{q_1,q_2})\odot(\bF_{N_p}^H\bd_{q_1,q_2})^* \right)}, \\
\bd_{q_1,q_2} &= \left[ \left(\tilde{\bx}_{q_1}\odot\tilde{\bx}_{q_2}^*\right)^T~ \bzero_{N_p-M_p,1}^T \right]^T .
\end{align} 
Combining (\ref{eq:Tf}) and (\ref{eq:Tf submatrix}), we have
\begin{align}
\bT_f &= \begin{bmatrix}
\bF_{N_p}^H \bSigma_{1,1} \bF_{N_p} &\cdots& \bF_{N_p}^H \bSigma_{1,Q} \bF_{N_p} \\
\vdots & \ddots & \vdots \\
\bF_{N_p}^H \bSigma_{Q,1} \bF_{N_p} &\cdots& \bF_{N_p}^H \bSigma_{Q,Q} \bF_{N_p} 
\end{bmatrix} \notag \\
&=(\bI_Q\otimes\bF_{N_p})^H\begin{bmatrix}
\bSigma_{1,1} & \cdots & \bSigma_{1,Q} \\
\vdots & \ddots & \vdots \\
\bSigma_{Q,1} & \cdots & \bSigma_{Q,Q} 
\end{bmatrix}(\bI_Q\otimes\bF_{N_p}).
\end{align}
By defining $\bSigma$ as
\begin{align}
\bSigma\triangleq \begin{bmatrix}
\bSigma_{1,1} & \cdots & \bSigma_{1,Q} \\
\vdots & \ddots & \vdots \\
\bSigma_{Q,1} & \cdots & \bSigma_{Q,Q} 
\end{bmatrix},
\end{align}
we then obtain
\begin{align}
\bT_f &= (\bI_{Q}\otimes\mathbf{F}_{N_p})^H\mathbf{\Sigma}(\bI_{Q}\otimes\mathbf{F}_{N_p}).
\end{align}
Thus, we conclude the proof.

\bibliographystyle{IEEEtran}
\bibliography{IEEEabrv,this_reference}

\end{document}